\documentclass[aps,pra,longbibliography,showpacs,notitlepage,twocolumn,superscriptaddress,nofootinbib,preprintnumbers]{revtex4-1}
\usepackage[ruled,vlined,linesnumbered]{algorithm2e}
\usepackage{bbm}
\usepackage{mathrsfs}
\usepackage{epsfig}
\usepackage{soul,xcolor}
\usepackage{graphicx}
\usepackage{amsfonts}
\usepackage{amsthm}
\usepackage[figuresright]{rotating}
\usepackage{amssymb}
\usepackage{amsmath}
\usepackage{dcolumn}
\usepackage{physics}
\usepackage{float}
\usepackage{bm}
\usepackage{physics}
\usepackage{verbatim}
\usepackage{braket}
\usepackage{cancel}

\usepackage[normalem]{ulem}
\usepackage{qcircuit}
\usepackage{enumitem}

\newtheorem{theorem}{Theorem}
\newtheorem{corollary}[theorem]{Corollary}
\newtheorem{lemma}[theorem]{Lemma}

\newtheorem*{theorem*}{Theorem}
\newtheorem*{corollary*}{Corollary}
\newtheorem*{lemma*}{Lemma}

\usepackage[colorlinks,linkcolor=blue,anchorcolor=blue,citecolor=blue,urlcolor=blue]{hyperref}

\begin{document}
\title{Halving the Cost of Quantum Algorithms with Randomization}

\author{John M. Martyn} 
\affiliation{Center for Theoretical Physics, Massachusetts Institute of Technology, Cambridge, Massachusetts 02139, USA}
\affiliation{IBM Quantum, MIT-IBM Watson AI Lab, Cambridge, Massachusetts 02142, USA}

\author{Patrick Rall} 
\affiliation{IBM Quantum, MIT-IBM Watson AI Lab, Cambridge, Massachusetts 02142, USA}

\begin{abstract}
    Quantum signal processing (QSP) provides a systematic framework for implementing a polynomial transformation of a linear operator, and unifies nearly all known quantum algorithms. In parallel, recent works have developed \emph{randomized compiling}, a technique that promotes a unitary gate to a quantum channel and enables a quadratic suppression of error (i.e., $\epsilon \rightarrow O(\epsilon^2)$) at little to no overhead. Here we integrate randomized compiling into QSP through \emph{Stochastic Quantum Signal Processing}. Our algorithm implements a probabilistic mixture of polynomials, strategically chosen so that the average evolution converges to that of a target function, with an error quadratically smaller than that of an equivalent individual polynomial. Because nearly all QSP-based algorithms exhibit query complexities scaling as $O(\log(1/\epsilon))$---stemming from a result in functional analysis---this error suppression reduces their query complexity by a factor that asymptotically approaches $1/2$. By the unifying capabilities of QSP, this reduction extends broadly to quantum algorithms, which we demonstrate on algorithms for real and imaginary time evolution, phase estimation, ground state preparation, and matrix inversion.
\end{abstract}

\preprint{MIT-CTP/5756}
\maketitle

\section{Introduction} 

Classical randomness plays a pivotal role in the design of quantum protocols and algorithms. In the near-term, randomized benchmarking~\cite{Knill_2008} is central to calibrating and assessing the quality of quantum gates, and quasi-probability methods like probabilistic error cancellation and noise twirling can help reduce noise~\cite{van_den_Berg_2023}. Similarly, random circuit sampling is central to quantum supremacy experiments \cite{Arute_2019}, and randomized measurements provide a powerful probe into the properties of complex quantum systems~\cite{Elben_2022}. As we progress towards quantum advantage and early fault-tolerant quantum hardware, many lines of research aim to reduce the requirements of traditional quantum algorithms by incorporating classical randomness~\cite{lin2020optimal, Dong_2022, lin2022heisenberg}.

\emph{Randomized compiling} is a key example of leveraging classical randomness to improve quantum computation~\cite{Wallman_2016, Campbell_2017, hastings2016turning}. As its name suggests, this process randomly compiles gates at execution time, or equivalently, promotes a unitary gate to a quantum channel that is a probabilistic mixture of unitaries. Remarkably, randomized compiling can quadratically suppresses gate errors without increasing the cost of circuit synthesis. Yet, applications of this technique to quantum algorithms have so far been restricted to Trotterized Hamiltonian simulation \cite{Campbell_2019, Childs_2019, Ouyang_2020, cho2022doubling, nakaji2023qswift} and phase estimation~\cite{Wan_2022}, leaving a vacuum of applications to other algorithms.  

In an effort to fill this gap of randomized quantum algorithms, we propose using quantum signal processing (QSP)~\cite{Low_2016, Low_2017} as a medium for achieving widespread advantage of randomized compiling. QSP prepares a polynomial transformation of a linear operator, and has been shown to encompass nearly all quantum algorithms, from Hamiltonian simulation and quantum search, to matrix inversion and fast integer factoring \cite{martyn2021grand, Gily_n_2019}.

In this work, we achieve exactly this goal. We merge randomized compiling with QSP by developing \emph{Stochastic Quantum Signal Processing} (Stochastic QSP). By virtue of randomized compiling, our construction quadratically suppresses the error in a QSP polynomial approximation of a target function. To study how this suppression impacts the cost of QSP-based algorithms, we show that an elementary result in the approximation of smooth functions implies that nearly all QSP-based algorithms achieve a query complexity that scales with the error $\epsilon$ as $O(\log(1/\epsilon))$, which we also empirically confirm. Hence the quadratic suppression of error afforded by stochastic QSP translates to an asymptotic halving of the cost of a QSP-based algorithm over their deterministic counterparts (asymptotic in the limit of $\log(1/\epsilon)$ dominating the cost). In realizing this cost reduction, we ``combine the strengths of QSP and randomization," as hypothesized in Ref.~\cite{berry2019random}.

An outline of this work is as follows. After reviewing QSP and other preliminary topics in Sec.~\ref{sec:preliminaries}, we present stochastic QSP in Sec.~\ref{sec:sqsp}. We then demonstrate the versatility of our scheme by showing its compatibility with several generalizations and variants of QSP in Sec.~\ref{sec:Generalizations}. Finally, in Sec.~\ref{Sec:Applications}, we benchmark the performance of stochastic QSP for various polynomials relevant to quantum algorithms, including those for real and imaginary time evolution, phase estimation, ground state preparation, and matrix inversion. 
In Appendix~\ref{app:SmoothFunctions} we review some useful results on Fourier and Chebyshev expansions of smooth functions, and in Appendix~\ref{app:Mixing_Lemmas} we prove some extensions of randomized compilation techniques.

\section{Preliminaries}\label{sec:preliminaries}
We begin by discussing the preliminary topics of this work, including notation and background concepts (Sec.~\ref{sec:Notation}), quantum signal processing (QSP) (Sec.~\ref{sec:QSP}), polynomial approximations to smooth functions (Sec.~\ref{sec:PolyApproxSmooth}), and randomized compiling (Sec.~\ref{sec:RandCompMixingLemma}).

\subsection{Notation and Background Concepts}\label{sec:Notation}

We will study functions $F(x)$ on the domain $x \in [-1,1]$ and use the function norm
\begin{equation}
    \| F \|_{[-1,1]} := \max_{x\in[-1,1]} |F(x)| , 
\end{equation}
where we will focus on functions bounded as $\|F \|_{[-1,1]} \leq 1$. 

A convenient set of functions on this domain are the Chebyshev polynomials. The order $n$ Chebyshev polynomial is defined as $T_n(x) = \cos(n \arccos(x))$ for integer $n\geq 0$, and is a polynomial of degree $n$ with parity $n \bmod 2$ (i.e., either even or odd) and bounded magnitude $\|T_n\|_{[-1,1]} = 1$. The Chebyshev polynomials furnish an orthogonal basis in which an arbitrary function on $x\in [-1,1]$ can be expanded: 
\begin{equation}\label{eq:ChebyExpansion}
    F(x) = \frac{c_0}{2} + \sum_{n=1}^\infty c_n T_n(x) , \quad  c_n = \frac{2}{\pi} \int_{-1}^1 \frac{F(x) T_n(x)}{\sqrt{1-x^2}} dx , 
\end{equation}
where $c_n$ are the Chebyshev coefficients.

In this work we will also study unitary and non-unitary transformations. We will denote an operator by a Latin character, say $A$, and the associated channel by the corresponding calligraphic character: $\mathcal{A}(\rho) = A \rho A^\dag$. In analyzing these operators, we will consider the spectral norm $\|A \| = \sup_{|\psi \rangle} \left\| A |\psi \rangle \right\|$ and the trace norm $\| A \|_1 = \tr \big( \sqrt{A^\dag A} \big) $, which equate to the maximal singular value and the sum of singular values, respectively.

Another relevant metric is the diamond norm, defined for a channel $\mathcal{E}$ as:
\begin{equation}
    \| \mathcal{E} \|_\diamond = \sup_{\rho} \left\| (\mathcal{E} \otimes \mathcal{I})(\rho) \right\|_1 , 
\end{equation}
where this supremum is taken over normalized density matrices $\rho$ in a possibly-enlarged Hilbert space, and $\mathcal{I}$ is the identity channel. The diamond norm induces the diamond distance between two channels:
\begin{equation}\label{eq:diamond_distance}
\begin{aligned}
    d_\diamond (\mathcal{E}, \mathcal{F})  :&= \frac{1}{2} \| \mathcal{E} - \mathcal{F}\|_\diamond . 
\end{aligned}
\end{equation}
For channels $\mathcal{A}(\rho) = A \rho A^\dag$ and $\mathcal{B}(\rho) = B \rho B^\dag$ with spectral norms $\|A\|, \|B \| \leq 1$, the diamond distance is upper bounded as (see Lemma 4 of Ref.~\cite{Rall_2021} for proof):
\begin{equation}\label{eq:diamond_distance_bound}
    d_\diamond (\mathcal{A}, \mathcal{B}) = \frac{1}{2}  \| \mathcal{A} - \mathcal{B}\|_\diamond \leq \| A - B \| . 
\end{equation}

\subsection{Quantum Signal Processing}\label{sec:QSP}
The quantum signal processing (QSP) algorithm is a systematic method of implementing a polynomial transformations on a quantum subsystem~\cite{Low_2016, Low_2017, Low_2019}. QSP works by interleaving a signal operator $U$, and a signal processing operator $S$, both taken to be SU(2) rotations about different axes. Conventionally, $U$ is an $x$-rotation through a fixed angle, and the $S$ a $z$-rotation through a variable angle $\phi$:
\begin{equation}\label{eq:QSP_inputs}
    U(x) = \begin{bmatrix}
        x & i\sqrt{1-x^2} \\
        i\sqrt{1-x^2} & x
    \end{bmatrix}, \qquad 
    S(\phi) = e^{i\phi Z}.
\end{equation}
Then, with a set of $d+1$ \textit{QSP phases} $\vec{\phi} = (\phi_0, \phi_1, ... , \phi_d) \in \mathbb{R}^{d+1}$, the following \textit{QSP sequence} is defined as an interleaved product of $U$ and $S$, whose matrix elements are polynomials in $x$:
\begin{equation}\label{eq:QSP_seqeunce}
    \begin{aligned}
        U_{\vec{\phi}}(x) &= S(\phi_0) \prod_{i=1}^d  U(x)  S(\phi_i) \\ 
        &= \begin{bmatrix}
        P(x) & iQ(x)\sqrt{1-x^2} \\
        iQ^*(x)\sqrt{1-x^2} & P^*(x)
        \end{bmatrix}, 
    \end{aligned}
\end{equation}
where $P(x)$ and $Q(x)$ are polynomials parameterized by $\vec{\phi}$ that obey:
\begin{equation}\label{eq:qsp_conditions}
    \begin{split}
        & \text{1. } {\rm deg}(P) \leq d, \ {\rm deg}(Q) \leq d-1 \\
        & \text{2. } P(x)\ \text{has parity } d \bmod 2, \text{ and } Q(x)\ \text{has parity } \\
        & \ \ \ (d-1)\bmod 2 \\
        & \text{3. } |P(x)|^2 + (1-x^2) |Q(x)|^2 = 1, \ \forall \ x \in [-1,1]. 
    \end{split}
\end{equation}

This result implies that one can prepare polynomials in $x$ by projecting into a block of $U_{\vec{\phi}}$, e.g. $\langle 0| U_{\vec{\phi}} | 0\rangle = P(x)$. While this class of polynomials is limited by the conditions of Eq.~\eqref{eq:qsp_conditions}, one can show that by projecting into other bases (e.g., the $| + \rangle \langle + |$ basis), and incorporating linear-combination-of-unitaries circuits~\cite{Berry_2014, Berry_2015_Simulating, Berry_2015_Hamiltonian}, QSP can encode an arbitrary degree-$d$ polynomial that need only obey $\|P\|_{[-1,1]} \leq 1$~\cite{Gily_n_2019}. For any such polynomial, the corresponding QSP phases $\vec{\phi}$ can be efficiently determined classically~\cite{Haah_2019, chao2020finding, dong2021efficient, Ying_2022, yamamoto2024robust}, thus amounting to merely a pre-computation step. As per Eq.~\eqref{eq:QSP_seqeunce}, the cost of realizing such a degree-$d$ polynomial is $d$ queries to $U(x)$.

Remarkably, QSP can be generalized to implement polynomial transformations of linear operators through its extension to the quantum eigenvalue transformation (QET)~\cite{Low_2017, Low_2019} and quantum singular value transformation (QSVT)~\cite{Gily_n_2019}. This is achieved analogous to QSP: provided access to a unitary $U[A]$ that block-encodes an operator $A$, one can construct a sequence of $U[A]$ and parameterized rotations that encodes a polynomial $P(A)$:
\begin{equation}\label{eq:QET_QSVT_seq}
    U[A]=\begin{bmatrix}
        A & \cdot \\
        \cdot & \cdot
    \end{bmatrix} \ \mapsto \ U_{\vec{\phi}}[A] = 
    \begin{bmatrix}
        P(A) & \cdot \\
        \cdot & \cdot
    \end{bmatrix} ,
\end{equation}
where the unspecified entries ensure unitarity. In essence, this applies QSP within each eigenspace (or singular value space) of $A$, and outputs a degree-$d$ polynomial $P(A)$ acting on the eigenvalues (or singular values) of $A$. The cost of realizing this is $d$ queries to the block-encoding, translating to a runtime $O(d)$.

QET and QSVT are powerful algorithms, shown to unify and simplify most known quantum algorithms, while maintaining near-optimal query complexities~\cite{Gily_n_2019}. An algorithm can be cast into the language of QET/QSVT by constructing a polynomial approximation to a matrix function that solves the problem of interest. For instance, in Hamiltonian simulation, one can design a polynomial $P(H) \approx e^{-iHt}$ to simulate time evolution~\cite{Low_2017, Low_2019}. Algorithms encompassed in this framework include the elementary algorithms of search, simulation, and phase estimation~\cite{martyn2021grand, Rall_2021}, as well as more intricate algorithms, like matrix inversion~\cite{lin2020optimal, tong2021fast, dalzell2024shortcut}, and ground state preparation~\cite{lin2022heisenberg, Dong_2022}.

In this work, we use the term``QSP" in place of ``QET" and ``QSVT", following the conventional parlance. However, this should be understood to be QET/QSVT when acting on a linear operator rather than a scalar.

\subsection{Polynomial Approximations to Smooth Functions}\label{sec:PolyApproxSmooth}

As emphasized above, the utility of QSP lies in generating matrix functions without the need to unitarily diagonalize the underlying matrix. Specifically, QSP enables the approximation of a matrix function $F(A) \approx P(A)$ by a polynomial $P(A)$, where the accuracy of this approximation can be tuned by increasing the degree of $P(x)$. Because the cost of a QSP-based algorithms scales with the polynomial degree, their complexity rests on results in approximation theory. 

To make this connection concrete, consider the decaying exponential function $ e^{-\beta x}$ for a parameter $\beta>0$. In QSP, this function is employed to prepare thermal states and estimate partition functions at inverse temperature $\beta$~\cite{Gily_n_2019}. It is well established that $ e^{-\beta x}$ can be approximated to within additive error $\epsilon$ over $x\in[-1,1]$ by a polynomial of degree $d = O(\sqrt{\beta} \log(1/\epsilon))$~\cite{Gily_n_2019, low2017hamiltonian}. Similarly, consider the error function $\text{erf}(kx)$ for a parameter $k>0$, which is used in QSP to approximate the step function by selecting a large $k$ and is applicable to ground state preparation~\cite{Lin_2020, Dong_2022}. Prior work has proven that $\text{erf}(kx)$ can be approximated to within additive error $\epsilon$ by a polynomial of degree $d = O(k \log(1/\epsilon))$~\cite{Gily_n_2019, low2017hamiltonian}. In both cases, the degree grows with decreasing error and increasing parameters $\beta$ or $k$.

Observe that in each of these examples, the degree of the polynomial scales with the error as $O(\log(1/\epsilon))$. This scaling is a generic feature of polynomial approximations to smooth functions, which originates from the expansion of a function on $x\in[-1,1]$ in the basis of Chebyshev polynomials (see Eq.~\eqref{eq:ChebyExpansion}): 
\begin{equation}
    F(x) = \frac{c_0}{2} + \sum_{n=1}^\infty c_n T_n(x). 
\end{equation} 
As we prove in Appendix~\ref{app:SmoothFunctions}, if $F(x)$ is $C^\infty$ function (i.e., continuous and infinitely differentiable), then its Chebyshev coefficients decay super-polynomially as $|c_n| = e^{-O(n^r)}$ for some exponent $r > 0$.

For a large class of smooth functions, it is found that $r=1$~\cite{boyd2001chebyshev}, such that $|c_n| = e^{-O(n)}$ decays geometrically. In this case, a truncation of the Chebyshev series at order $d$, $P(x) = \sum_{n=0}^{d} c_n T_n(x)$, furnishes a degree $d$ polynomial approximation to $F(x)$ that suffers error 
\begin{equation}
\begin{aligned}
    &\max_{x\in [-1,1]} \big| P(x) - F(x) \big| = \max_{x \in [-1,1]} \Bigg|\sum_{n=d+1}^\infty c_n T_n(x) \Bigg| \\ 
    & \  \leq \sum_{n=d+1}^\infty |c_n| = \sum_{n=d+1}^\infty e^{-O(n)} = e^{-O(d)}.
\end{aligned}
\end{equation}
Hence, to guarantee an error at most $\epsilon$, it suffices to choose a degree $d = O(\log(1/\epsilon))$.

In practice many polynomial approximations are constructed via truncated Chebyshev series. This includes polynomial approximations to a wide range of functions relevant to quantum algorithms, including the decaying exponential $e^{-\beta x}$, trigonometric functions $\sin(tx)$, $\cos(tx)$, the step function $\Theta(x)$, the inverse function $1/x$,\footnote{Although the step function $\Theta(x)$ and the inverse function $1/x$ exhibit singularities at $x=0$, and thus are not $C^\infty$ functions, they can however be approximated by $C^\infty$ functions by excluding a small region around their singularity. This strategy is used in practice, and renders these function amenable to results on polynomial approximations to smooth functions.} and beyond. Accordingly, these approximations all exhibit degrees that scale with the error as $d = O(\log(1/\epsilon))$, which carries over to the complexity of their corresponding QSP-based algorithms.

\subsection{Randomized Compiling and the Mixing Lemma}\label{sec:RandCompMixingLemma}

In order to incorporate randomization into QSP, we will use \emph{randomized compiling}. Formally introduced in Ref.~\cite{Wallman_2016}, randomized compiling replaces a deterministic quantum circuit with a circuit sampled from a distribution. This process can be viewed as replacing a unitary operation with a a quantum channel that is a probabilistic mixture of unitaries. 

Remarkably, if this mixture is chosen strategically, randomized compiling enables a quadratic suppression of error: if an individual gate approximates a target unitary with error $\epsilon$, the randomly compiled channel can approximate the corresponding target channel with error $O(\epsilon^2)$. This error suppression is achieved at little to no increase in overhead, requiring only the ability to classically sample a distribution and implement gates on the fly.

The precise error suppression is quantified by the Hastings-Campbell mixing lemma: 
\begin{lemma}[Hastings-Campbell Mixing Lemma~\cite{Campbell_2017, hastings2016turning}]\label{lemma:Mixing}
    Let $V$ be a target unitary operator, and $\mathcal{V}(\rho) = V \rho V^\dag $ the corresponding channel. Suppose there exist $m$ unitaries $\{ U_j \}_{j=1}^m$ and an associated probability distribution $p_j$ that approximate $V$ as
    \begin{equation}
    \begin{aligned}
        & \| U_j - V \| \leq a \text{ for all } j, \\ 
       & \Big\| \sum_{j=1}^m p_j U_j - V \Big\| \leq b ,
    \end{aligned}
    \end{equation}
    for some $a,b > 0$. Then, the corresponding channel $\Lambda(\rho) = \sum_{j=1}^m p_j U_j \rho U_j^\dag$ approximates $\mathcal{V}$ as
    \begin{equation}
        \| \Lambda - \mathcal{V} \|_\diamond \leq a^2 + 2b . 
    \end{equation}
\end{lemma}

This lemma enables a quadratic suppression of error if one can specify an ensemble of unitaries $\{ U_j \}$ that each achieve spectral error $a = \epsilon$, and a distribution $p_j$ such that $b=O(\epsilon^2)$. Then, while an individual unitary $U_j$ suffers diamond norm error $O(\epsilon)$, the channel $\Lambda$ achieves error $\leq a^2 + 2b = O(\epsilon^2)$. Importantly, because $\Lambda$ is a probabilistic mixture of the unitaries $\{ U_j \}$, the cost of simulating $\Lambda$ is no more expensive than the cost of sampling $p_j$ and implementing an individual unitary $U_j$.

The mixing lemma has been leveraged to improve a variety of quantum protocols through randomized compiling. Noteworthy examples include reducing the cost of gate synthesis~\cite{Campbell_2017, low2021halving, akibue_2024_unitary, yoshioka2024error, hastings2016turning}, tightening fault-tolerance thresholds for general noise models~\cite{Wallman_2016}, and enhancing the precision of state preparation~\cite{akibue2022quadratic, Akibue_2024_state}. On the algorithmic side, the mixing lemma has been merged with Trotterization to significantly reduce the complexity of chemical simulations~\cite{Campbell_2019, Ouyang_2020}, double the order of Trotter formulae~\cite{cho2022doubling}, and accelerate imaginary time evolution~\cite{Huang_2023, pocrnic2023composite}. Here we continue this campaign by extending randomized compiling to QSP. As QSP unifies nearly all quantum algorithms~\cite{martyn2021grand}, this paves the way for new randomized quantum algorithms with reduced query complexities.

\section{Stochastic Quantum Signal Processing}\label{sec:sqsp}
Our goal is to integrate randomized compiling into QSP, and thereby establish a framework for designing randomized quantum algorithms that achieve reduced query complexities. We begin in Sec.~\ref{sec:MixingLemmaBlockEncoding} by introducing an extension of the mixing lemma for operators block-encoded in unitaries. Then, in Sec.~\ref{sec:StochasticQSP_ssec}, we use this result to develop \emph{stochastic QSP}: we replace a single QSP polynomial with a channel that is a probabilistic mixture of QSP polynomials, each strategically crafted to exploit the mixing lemma and quadratically suppress error. As we show, this furnishes randomized QSP-based algorithms with roughly half the cost of their deterministic counterparts.

\subsection{The Mixing Lemma for Block-Encodings}\label{sec:MixingLemmaBlockEncoding}
As emphasized in Sec.~\ref{sec:QSP}, QSP polynomials are constructed as block-encodings. That is, a QSP polynomial $P(A)$ is encoded in a block of a higher dimensional unitary $U$, and accessed as $P(A) = \Pi U \Pi'$ for some orthogonal projectors $\Pi, \Pi'$. Conventionally, the projectors are taken to be $\Pi = \Pi' = |0\rangle \langle 0 | \otimes I$, such that $P(A)$ is encoded in the $|0\rangle \langle 0 |$ block of the unitary.

To apply the mixing lemma to QSP, it is therefore necessary to establish a variant of the mixing lemma for operators block-encoded in unitary transformations. For the sake of simplicity, we present this theorem for an operator encoded in the $|0\rangle \langle 0|$ block of a unitary:
\begin{lemma}[Mixing Lemma for Block-Encodings: $|0\rangle \langle 0|$ Block]\label{lemma:BlockEncodingMixing} 
    Let $V$ be a unitary that block-encodes a (possibly non-unitary) target operator $S$ as $S = (\langle 0 | \otimes I) V ( |0 \rangle \otimes I) $. Suppose there exist $m$ unitaries $\{ U_j \}_{j=1}^m$ that block-encode operators $R_j$ as $R_j = (\langle 0 | \otimes I)  U_j ( | 0 \rangle \otimes I)$. Also suppose there exists an associated probability distribution $p_j$ such that
    \begin{equation}
    \begin{aligned}
        &\| R_j  - S \| \leq a \text{ for all } j , \\ 
        &\Big\| \sum_{j=1}^m p_j R_j - S \Big\| \leq b .
    \end{aligned}
    \end{equation}
    Then, the corresponding unitary channel $\Lambda(\rho) = \sum_{j=1}^m p_j U_j \rho U_j^\dag$ approximates the action of the channel $\mathcal{V}$ as
    \begin{equation}
        \big\| \bar{\Lambda} - \bar{\mathcal{V}} \big\|_\diamond \leq a^2 + 2b , 
    \end{equation}
    where $\bar{\Lambda}$ and $\bar{\mathcal{V}}$ are channels that access the block-encodings of $\Lambda$ and $\mathcal{V}$ by appending an ancilla qubit and projecting onto $\ket{0}$:
    \begin{equation}
        \begin{aligned}
            & \bar{\Lambda}(\rho) = (\langle 0 | \otimes I) \cdot  \Lambda \big( | 0 \rangle \langle 0 | \otimes \rho \big) \cdot ( |0 \rangle \otimes I) \\
            & \bar{\mathcal{V}}(\rho) = (\langle 0 | \otimes I) \cdot \mathcal{V} \big( | 0 \rangle \langle 0 | \otimes \rho \big) \cdot ( |0 \rangle \otimes I).
        \end{aligned}
    \end{equation}
\end{lemma}
For brevity, we defer the proof of this theorem to Appendix~\ref{app:Mixing_Lemmas}; there, we also showcase an analogous result for arbitrary block-encodings accessed by projectors $\Pi, \Pi'$. Lemma~\ref{lemma:BlockEncodingMixing} indicates that by implementing the probabilistic mixture of block-encoding unitaries $\Lambda(\sigma) = \sum_{j=1}^m p_j U_j \sigma U_j^\dag$ on an input state $\sigma = |0\rangle \langle 0 | \otimes \rho$, and projecting the block-encoding qubit onto $|0\rangle $ (hence the projectors $| 0 \rangle \otimes I$ and $\langle 0 | \otimes I$), one can reproduce evolution under the target operator $S$ with diamond norm error $a^2 + 2b$. In short, this result is proven by showing that the channel $\bar{\Lambda}(\rho)$ is equal to the probabilistic mixture of block-encoded operators $\sum_{j=1}^m p_j R_j \rho R_j^\dag$, to which the mixing lemma applies. Parallel to the usual mixing lemma, this result enables a quadratic suppression of error by selecting operators $R_j$ and an associated probability distribution $p_j$ such that $a=\epsilon$ and $b = O(\epsilon^2)$.

\subsection{Stochastic QSP}\label{sec:StochasticQSP_ssec}
Lemma~\ref{lemma:BlockEncodingMixing} very naturally applies to QSP. In this context, the target operation is a matrix function: $S = F(A)$, yet the operators we have access to are QSP polynomials: $R_j = P_j(A)$. A common goal is to simulate evolution under $F(A)$ as $\mathcal{F}_A(\rho) = F(A)\rho F(A)^\dag$, which encompasses algorithms such as time evolution, linear systems solvers, and ground state preparation, among many others. Traditionally, one achieves this goal with QSP by finding a suitable polynomial approximation to $F(x)$ as $|P(x)-F(x)|\leq \epsilon$, such that evolving under this polynomial with QSP as $\mathcal{P}_A(\rho) = P(A) \rho P(A)^\dag$ suffers error $ \| \mathcal{P}_A - \mathcal{F}_A \|_\diamond \leq O(\epsilon) $. If $P(x)$ is a degree $d$ polynomial, this procedure requires $d$ queries to the block-encoding of $A$.

Here we exploit the mixing lemma to approximate evolution under $F(A)$ to the same level of accuracy, but at asymptotically half the number of queries to the block-encoding. We achieve this by designing an ensemble of polynomials that each approximate $F(A)$ as $\| P_j(A) - F(A)\| \leq O(\sqrt{\epsilon})$, and an associated probability distribution that obeys $\| \sum_j p_j P_j(A) - F(A) \| \leq O(\epsilon)$. Then, Lemma~\ref{lemma:BlockEncodingMixing} readily implies that the channel $\Lambda_A(\rho) = \sum_j P_j(A) \rho P_j(A)^\dag $ suffers error $\| \Lambda_A - \mathcal{F}_A \|_\diamond \leq O(\epsilon)$. We also show that implementing $\Lambda_A$ requires a number of queries to the block-encoding $\approx d/2 + O(1)$, a cost reduction stemming from the fact that polynomial approximations of smooth functions tend to have degrees that scale as $d = O(\log(1/\epsilon))$ (see Sec.~\ref{sec:PolyApproxSmooth}). Intuitively, this scaling implies that a polynomial that achieves error $O(\sqrt{\epsilon})$ (e.g., $P_j(x)$) has a degree that is asymptotically half that of a polynomial that achieves error $O(\epsilon)$ (e.g., $P(x)$).

\subsubsection{Presentation}

Thus, rather than implement a degree $d$ polynomial, one can instead sample over an ensemble of polynomials of \emph{average degree} $\approx d/2 + O(1)$, while retaining the same level of accuracy. As the corresponding channel is constructed as a probabilistic mixture of QSP sequences, we term this algorithm \emph{Stochastic Quantum Signal Processing}:

\begin{theorem}[Stochastic QSP] \label{thm:StochasticQSP} 
Suppose that $F(x)$ is a bounded function $\|F\|_{[-1,1]} \leq 1$ with a Chebyshev expansion
\begin{equation}
    F(x) = \sum_{n=0}^\infty c_n T_n(x)    
\end{equation}
on the domain $x \in [-1,1]$, where for some degree $d \geq 2$ the coefficients decay as
\begin{equation}\label{eq:poly_coef_decay_condition}
    |c_n| \leq  C e^{- q n} \ \text{for all} \ n \geq  d/2,
\end{equation}
for some constants $C, q > 0$. Suppose furthermore that a degree-$d$ truncation of $F(x)$, $P(x) = \sum_{n=0}^d c_n T_n(x)$, achieves an approximation error of $\epsilon$ as:
\begin{align}
\left| F(x) - P(x) \right| \leq \sum_{n=d+1}^\infty |c_n| =: \epsilon,
\end{align}
such that a QSP implementation of the channel $\mathcal{P}_A(\rho) = P(A) \rho P(A)^\dag$ for some operator $A$ deviates from the target channel $\mathcal{F}_A(\rho) = F(A) \rho F(A)^\dag$ by an error
\begin{equation}
    \| \mathcal{P}_A - \mathcal{F}_A \|_\diamond \leq 2\epsilon = O(\epsilon), 
\end{equation}
while making $d$ queries to the block-encoding of $A$. 

Then there exists an ensemble of polynomials $\{P_j(x)\}$ of degree $\mathrm{deg}(P_j) \leq d$, and an associated probability distribution $p_j$ such that:
\begin{equation}\label{eq:poly_ensemble_errors}
\begin{aligned}
&\left|  P_j(x)  - F (x) \right| \leq 2\sqrt{\epsilon}= O\big(\sqrt{\epsilon}\big) \\
&\Big| \sum_{j=1} p_j P_j(x)  - F(x) \Big| \leq \epsilon = O(\epsilon), 
\end{aligned}
\end{equation}
while the average degree of these polynomials is
\begin{equation}\label{eq:d_avg}
\begin{aligned}
    d_\text{avg} &:= \sum_{j=1} p_j \mathrm{deg}(P_j) \\
    &\leq \frac{d}{2} + \frac{\log(C)}{2q} - \frac{\log(1-e^{-q})}{2q} + \frac{1}{2} + \frac{1}{1-e^{-q}} \\
    & = \frac{d}{2} + O(1).
\end{aligned}
\end{equation}
Therefore, according to the mixing lemma for block-encodings (Lemma~\ref{lemma:BlockEncodingMixing}), the channel $\Lambda_A(\rho) = \sum_{j=1} p_j P_j(A) \rho P_j(A)^\dag$ suffers error
\begin{equation}
    \|\Lambda_A - \mathcal{F}_A \|_\diamond \leq 6 \epsilon = O(\epsilon),
\end{equation}
while making $d_\text{avg} \leq d/2 + O(1)$ queries to a block-encoding of $A$ in expectation. The cost reduction realized by this channel is~\footnote{Here, $\lesssim$ neglects the terms independent of $d$ and $C$ in Eq.~\eqref{eq:d_avg}, which are less relevant than $\log(C)/2q$ in practice. See examples in Sec.~\ref{Sec:Applications}.} 
\begin{equation}
    \frac{d_{\text{avg}}}{d} \lesssim \frac{1}{2}\Big(1+ \frac{\log(C)}{qd} \Big) ,
\end{equation}
which approaches $1/2$ in the limit of large $d$. 
\end{theorem}

\begin{proof}
    The underlying goal of this theorem is to simulate evolution under $F(A)$ according to the channel $\mathcal{F}_A(\rho) = F(A) \rho F(A)^\dag$. For notational convenience, let us define the degree $t$ polynomial truncation of the Chebyshev series of $F(x)$ as
    \begin{equation}
        P^{[t]}(x) := \sum_{n=0}^{t} c_n T_n(x).
    \end{equation}
    By the assumption that $|c_n| \leq C e^{-qn} $ for $n \geq d/2 $, we find that for $t \geq d/2 $, this polynomial truncation suffers an error over $x \in [-1,1]$:
    \begin{equation}\label{eq:poly_truncation_error}
    \begin{aligned}
        & \big| P^{[t]}(x) - F(x) \big| = \Bigg|\sum_{n=t+1}^\infty c_n T_n(x) \Bigg| \\
        &\leq \sum_{n=t+1}^\infty |c_n| \leq  \sum_{n=t+1}^\infty C e^{-qn} = \frac{Ce^{-qt}}{1-e^{-q}} .  
    \end{aligned}
    \end{equation}
    Therefore, the error suffered by the degree $d$ truncation $P^{[d]}(x) = P(x)$ is at most
    \begin{equation}\label{eq:epsilon_expression}
        |P^{[d]}(x) - F(x) | \leq \frac{Ce^{-qd}}{1-e^{-q}} =: \epsilon.
    \end{equation}
    
    To estimate evolution under $\mathcal{F}_A$, one can use QSP to implement the channel $\mathcal{P}_A(\rho) = P^{[d]}(A) \rho P^{[d]}(A)^\dag$. Using the diamond norm bound of Eq.~\eqref{eq:diamond_distance_bound}, we see that $\mathcal{P}_A$ suffers error 
    \begin{equation}
    \begin{aligned}
        \| \mathcal{P}_A - \mathcal{F}_A \|_\diamond & \leq 2 \| P^{[d]}(A) - F(A) \| \\
        & \leq 2 \max_{x\in[-1,1]} \big| P^{[d]}(x) - F(x) \big| \\
        & \leq 2 \epsilon. 
    \end{aligned}
    \end{equation}
    As $P^{[d]}(A)$ is a degree $d$ polynomial, the cost of implementing $\mathcal{P}_A$ is $d$ queries to the block-encoding of $A$. 

    Our goal is to reproduce evolution under $\mathcal{F}_A$ to diamond norm error $O(\epsilon)$ by using a probabilistic mixture of QSP polynomials and invoking the mixing lemma. We construct these polynomials by truncating the original series at a cutoff degree $d^*$, selected as follows. Because the mixing lemma quadratically suppresses error, and the degree scales as $\log(1/\epsilon)$, there should exist an ensemble of polynomials of degree around $d^* \approx d/2$ whose corresponding channel suffers the same error as the original degree $d$ polynomial. We can then readily determine $d^*$ by demanding that it be the smallest integer suffering error at most $\sqrt{\epsilon}$:
    \begin{equation}
    \begin{aligned}
        & \Bigg( \frac{Ce^{-qd^*}}{1-e^{-q}} \Bigg)^2 \leq \epsilon = \frac{Ce^{-qd}}{1-e^{-q}} \quad \Rightarrow \\
        & d^* = \Bigg\lceil \frac{d}{2} + \frac{\log(C)}{2q} - \frac{\log(1-e^{-q})}{2q} \Bigg\rceil \\
        & \ \ \  = \frac{d}{2} + O(1).
        \label{eqn:dstarselect}
    \end{aligned}
    \end{equation}
    Because the error suffered by a single polynomial obtained by truncating at $d^*$ is greater than that at $d$, we have $d^* < d$. 

    Next, let us define our ensemble of polynomials as
    \begin{equation}\label{eq:polynomial_ensemble}
    \begin{aligned}
        &P_j(x) = 
        \begin{cases}
            P^{[ d^* ]}(x) + \frac{c_{d^*+j}}{p_j} T_{d^* +j}(x) & |c_{d^* + j}| > 0 \\ 
            0 & c_{d^* +j} =0
        \end{cases} 
    \end{aligned}
    \end{equation}
    for $j=1,2,..., d-d^*$, where $p_j$ is the associated probability distribution defined as
    \begin{equation}
        p_j = \frac{|c_{ d^* +j}|}{\sum_{k=1}^{d-d^* } |c_{d^* +k}|}. \label{eqn:probabilities}
    \end{equation}
    Intuitively, each polynomial $P_j(x)$ consists of the degree $d^* $ truncation $P^{[ d^* ]}(x)$, and an additional higher order Chebyshev polynomial chosen such that the average polynomial is the degree $d$ truncation: 
    \begin{equation}
        \sum_{j=1}^{ d-d^* } p_j P_j(x) = P^{[d]}(x).
    \end{equation} 
    The distribution $p_j$ is chosen such that terms with larger Chebyshev coefficients are given more probability mass and are preferentially sampled. Each polynomial in this ensemble is guaranteed to suffer error
    \begin{equation}
    \begin{aligned}
        |P_j(x) - F(x)| & \leq |P^{[ d^* ]}(x) - F(x) | + \frac{|c_{ d^* +j}|}{p_j} \\
        &\leq \sum_{n= d^* +1}^{\infty} |c_{n}| + \sum_{k=1}^{d-d^*} |c_{ d^* +k}|\\
        & \leq 2 \sum_{n=d^* +1}^{\infty} |c_n| \ \leq \  2 \sum_{n=d^* +1}^{\infty} C e^{-qn} \\
        & = \frac{2Ce^{-q d^* }}{1-e^{-q}} \leq 2\sqrt{\epsilon} = O(\sqrt{\epsilon}). 
    \end{aligned}
    \end{equation}
    This bound also implies that  $P_j(x)$ is bounded as $\| P_j \|_{[-1,1]} \leq 1+2\sqrt{\epsilon} $. On the other hand, the average polynomial suffers error
    \begin{equation}
        \Bigg| \sum_{j=1}^{ d-d^*} p_j P_j(x) - F(x) \Bigg| = \big| P^{[d]}(x) - F(x) \big| \leq \epsilon.
    \end{equation}

    In the language of the mixing lemma for block-encodings (Lemma~\ref{lemma:BlockEncodingMixing}), this corresponds to values $a = 2\sqrt{\epsilon} $ and $b = \epsilon$. Accordingly, the channel $\Lambda_A(\rho) = \sum_{j=1}^{d-d^*} p_j P_j(A) \rho P_j(A)^\dag $ suffers error
    \begin{equation}
        \| \Lambda_A - \mathcal{F}_A \|_\diamond \leq a^2 + 2b = 6 \epsilon = O(\epsilon) .
    \end{equation}

    Lastly, the expected cost of instantiating $\Lambda_A(\rho)$ is the average degree $d_{\text{avg}} = \sum_{j=1}^{ d-d^* } p_j \deg(P_j)$, which corresponds to the average number of queries to the block-encoding. Note that the degrees of these polynomials are $\deg(P_j) = d^* +j \leq d$. To evaluate the average degree, recall that this theorem assumes that 
    \begin{equation}
        |c_n| \leq  C e^{- q n} \ \text{for all} \ n \geq d/2 ,
    \end{equation}
    This implies that the mean of the distribution $p_j := \frac{|c_{d^* +j}|}{\sum_{k=1}^{d-d^* } |c_{d^* +k}|} $ is upper bounded by the mean of the geometric distribution $\tilde{p}_j := \frac{e^{-q(d^* +j)}}{\sum_{k=1}^{d-d^* } e^{-q( d^* +k)} } = \frac{e^{-qj}}{ \sum_{k=1}^{d-d^*} e^{-q k} }$. Hence, we may upper bound $d_{\text{avg}}$ as  
    \begin{equation}
    \begin{aligned}
        d_{\text{avg}} &= \sum_{j=1}^{ d-d^*} p_j \deg(P_j) = d^* + \sum_{j=1}^{ d-d^*} j p_j \\
        &\leq d^* + \sum_{j=1}^{ d-d^*} j \tilde{p}_j \\
        & = d^* - (d-d^*) \frac{e^{-q (d-d^*)}}{1-e^{-q (d-d^*) }} + \frac{1}{1-e^{-q}} \\
        & \leq d^* + \frac{1}{1-e^{-q}} \\
        & = \Bigg\lceil \frac{d}{2} + \frac{\log(C)}{2q} - \frac{\log(1-e^{-q})}{2q} \Bigg\rceil + \frac{1}{1-e^{-q}} \\
        & \leq \frac{d}{2} + \frac{\log(C)}{2q} - \frac{\log(1-e^{-q})}{2q} + \frac{1}{2} + \frac{1}{1-e^{-q}} \\
        & = \frac{d}{2} + O(1)
    \end{aligned}
    \end{equation}
    where line 3 follows from evaluating mean of the geometric distribution (i.e., evaluating a geometric series), and line 4 from $d^* < d$. Accordingly, the cost reduction realized by $\Lambda_A$ is the ratio 
    \begin{equation}
    \begin{aligned}
        \frac{d_{\text{avg}}}{d} &\leq \frac{1}{2}\Bigg(1+ \frac{\log(C)}{qd} - \frac{\log(1-e^{-q})}{qd} + \frac{1}{d} + \frac{2}{(1-e^{-q}) d }  \Bigg) \\ \label{eqn:davgratio}
        & \approx \frac{1}{2}\Big(1+ \frac{\log(C)}{qd} \Big)
    \end{aligned}
    \end{equation}
    where the last line follows from the fact that the $\log(C)$ contribution dominates in practice (see Sec.~\ref{Sec:Applications}). In the large $d$ limit, the cost reduction approaches $1/2$ inverse-polynomially fast.  
    
    Lastly, as this construction makes no reference to the eigenvalues or singular values of $A$, stochastic QSP applies equally as well to QET and QSVT, where $P(A)$ acts on the eigenvalues or singular vectors of $A$, respectively. In addition, while the presentation of stochastic QSP here is tailored toward functions expressed in the basis of Chebyshev polynomials, we extend this result to more general functions and arbitrary arbitrary bases in Sec.~\ref{sec:Generalizations}.
\end{proof}

\subsubsection{Interpretation}

Let us take a minute to interpret this result. According to Theorem~\ref{thm:StochasticQSP}, stochastic QSP replaces a deterministic polynomial $P(x)$ with an ensemble of polynomials $\{ P_j(x) \}$ and probability distribution $p_j$, whose average evolution achieves the same precision up to a constant factor, but at asymptotically half the cost. Importantly, this result is agnostic to the specific polynomial, requiring only that its coefficients decay exponentially according to Eq.~\eqref{eq:poly_coef_decay_condition}. As we showed in Sec.~\ref{sec:PolyApproxSmooth}, this condition is generally satisfied by polynomial approximations to smooth functions, rendering stochastic QSP applicable to a wide range of algorithms. 
For visual intuition on this exponential decay, we provide an illustration of our stochastic QSP construction in Fig.~\ref{fig:Stochastic_QSP_Illustration}.
In practice the values of $C$ and $q$ in Eq.~\eqref{eq:poly_coef_decay_condition} can be chosen to minimize the ratio $d_\text{avg}/d$, which effectively means minimizing $\log(C)/q$.

\begin{figure}
    \centering
    \includegraphics[width=0.98\linewidth]{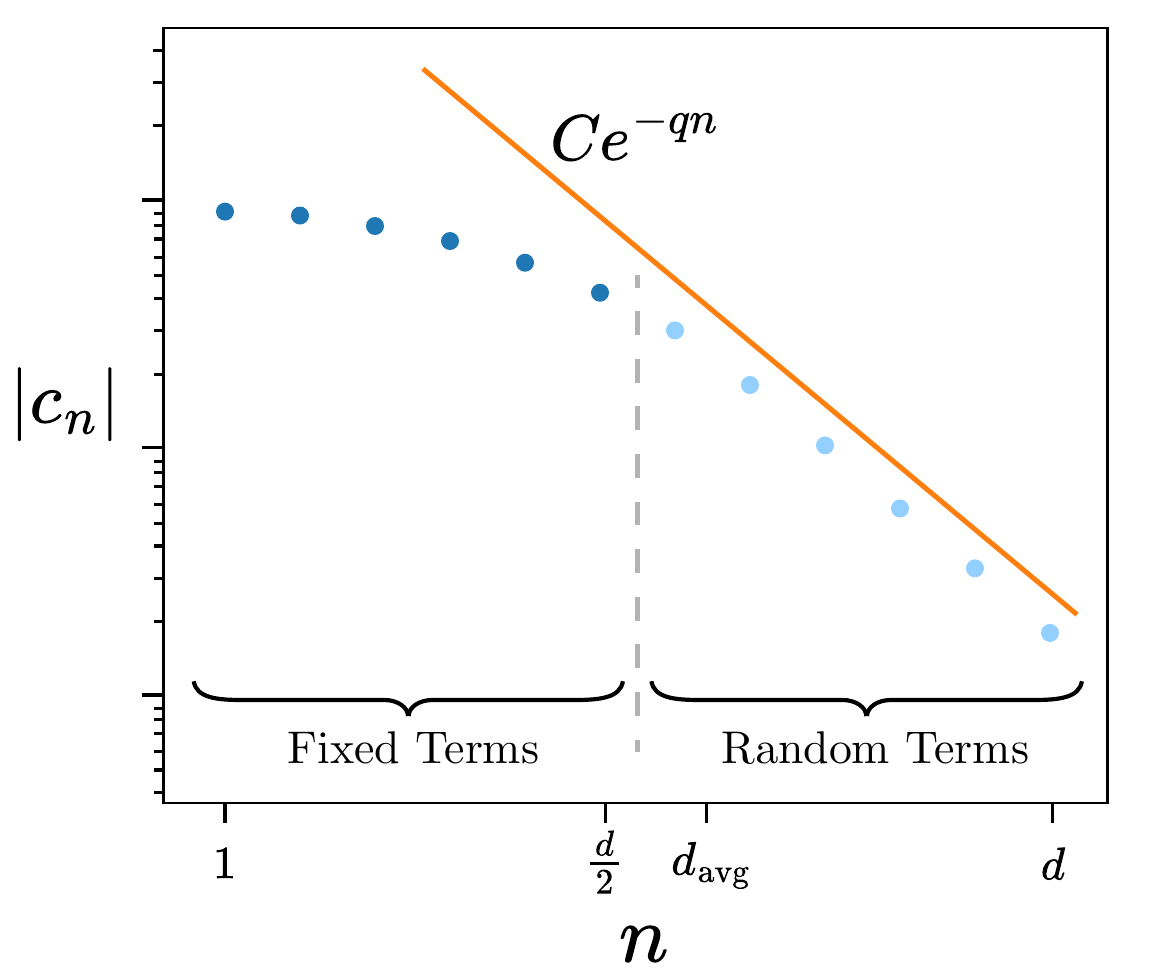}
    \caption{An illustration of our stochastic QSP construction. According to Theorem~\ref{thm:StochasticQSP}, we are provided a function whose Chebyshev coefficients $c_n$ decay exponentially as $|c_n| \leq Ce^{-qn}$ for $n\geq d/2+1$. Then, by defining a cutoff degree $\approx d/2$ (see Eq.~\eqref{eqn:dstarselect}), we can design an ensemble of polynomials that consists of all lower order terms $n\leq d/2$, and a single higher order term $d/2<n\leq d$ that is randomly sampled (see Eq.~\eqref{eq:polynomial_ensemble}). The average degree of these polynomials is $d_{\text{avg}} = \frac{d}{2}+O(1)$.}
    \label{fig:Stochastic_QSP_Illustration}
\end{figure}

The channel implemented by stochastic QSP is the probabilistic mixture of polynomials $\Lambda_A(\rho) = \sum_j p_i P_j(A) \rho P_j(A)$. As we discussed in Sec.~\ref{sec:MixingLemmaBlockEncoding}, this channel may be realized by implementing an identical probabilistic mixture of unitaries $\{ U_j \}$ that block-encode the polynomials $\{ P_j(A) \}$, and post-selecting on successfully accessing these block-encodings. In practice, this can be achieved by independently sampling $j \sim p_j$, and implementing the QSP sequence that block-encodes the polynomial $P_j(A)$. Because $P_j(x)$ is only bounded as $\| P_j \|_{[-1,1]} \leq 1+2\sqrt{\epsilon} $ according to Eq.~\eqref{eq:poly_ensemble_errors}, this implementation may require rescaling the polynomials by $1+2\sqrt{\epsilon}$, which incurs a measurement overhead $\sim (1+2\sqrt{\epsilon})^2$ that asymptotically approaches $1$.

This procedure of course requires knowledge of the QSP phases for each polynomial $P_j(A)$. These phases can be determined classically in a pre-computation step by using an efficient phase finding algorithm, such as those of Refs.~\cite{Haah_2019, chao2020finding, Dong_2021, Ying_2022, yamamoto2024robust}, and the associated circuits can similarly be pre-compiled. Because the ensemble consists of $\approx d/2$ polynomials, each of degree at most $d$, the total storage required is $O(d^2)$ QSP phases. 
%Consequently, stochastic QSP requires such a classical pre-computation step, with storage and computation cost $\text{poly}(d)$, as in ordinary QSP. 

Moreover, observe from Eq.~\eqref{eq:polynomial_ensemble} that each polynomial in the ensemble is constructed as the degree $d^* \approx d/2  $ truncation plus a higher order term in the Chebyshev expansion, up to degree $d$. Inclusion of the degree $d^*$ truncation guarantees that the first condition of the mixing lemma is satisfied with error $O(\sqrt{\epsilon})$, and sampling the higher order terms allows the second condition of the mixing lemma to be satisfied with error $O(\epsilon)$. The specific sampling distribution is chosen to be proportional to the coefficients of the Chebyshev expansion. Because these are assumed to decay exponentially, the corresponding probability mass is concentrated around $\approx d/2$, and so too is the average degree. Through this view, the ensemble of stochastic QSP can be seen as a fixed low order term plus higher order correction terms, which are importance sampled to leverage randomized compiling. 
%An ensemble with a similar structure, albeit in the context of Trotterization, was recently used in Ref.~\cite{cho2022doubling} to double the order of Trotter formulae through randomized compiling. 

Lastly, it is important to note that while stochastic QSP reduces the expected cost to $d_{\text{avg}} \approx d/2 +O(1)$, it does not reduce the maximum degree of the polynomials implemented: some polynomials in the ensemble will have degree greater than $d_{\text{avg}}$, with one even having degree $d$. This however is unavoidable. In fact it is necessary for the ensemble to contain polynomials of degree $> d/2$ in order to attain a level of error equivalent to a degree-$d$ polynomial. To see this, observe that if all the polynomials in the ensemble had degree at most $k < d$, then their average $\sum_j p_j P_j(x)$ would also be a degree $k<d$ polynomial. This average polynomial would not be able to achieve a precision equivalent to that of a degree $d$ polynomial, thus failing to meet the second condition of the mixing lemma. Note however that the distribution $p_j$ concentrates around small values of $j$, meaning that degrees much larger than $\approx d/2$ are rare.

Through this interpretation, stochastic QSP is similar to the ``semi-quantum matrix processing" algorithm of Ref.~\cite{tosta2023randomized} for estimating expectation values and matrix elements of a matrix-valued function. The authors achieve this by decomposing a degree-$d$ polynomial approximation of this function into the Chebyshev basis, and sampling its constituent Chebsyhev polynomials, and measuring an estimator of the sampled polynomial that converges to the desired expectation value/matrix element. Like stochastic QSP, this procedure can reduce the expected cost of certain algorithms to a value $< d$ if the Chebyshev coefficients decay quickly, but it does not reduce the maximal degree because the order $d$ Chebyshev polynomial can still be sampled. However, this method differs from stochastic QSP in two key ways. First, Ref.~\cite{tosta2023randomized} employs a quasi-probability technique where improvements in average degree are traded for additional variance in the measurements, whereas stochastic QSP always implements a normalized probability distribution. Secondly, the constructions in Ref.~\cite{tosta2023randomized} explicitly consider a target observable $O$ and modify their circuit to a accommodate a parity measurement $Z \otimes O$ with an ancilla register. Hence, their construction does not constitute an approximate realization of a desired matrix function, whereas stochastic QSP indeed approximates evolution under said matrix function, thus rendering our approach more modular.

\subsubsection{Utility in Quantum Algorithms}
Stochastic QSP can be integrated into larger quantum algorithms by replacing a deterministic QSP operation with the corresponding stochastic QSP channel. This remains true despite the stochastic QSP channel being incoherent, as it is a probabilistic mixture of operations.

As we prove in Theorem~\ref{thm:StochasticQSP}, the stochastic QSP channel closely approximates the target channel in the diamond norm, where the target channel is coherent. Therefore, the stochastic QSP channel can be substituted in for the target channel in a larger quantum algorithm, incurring only an $O(\epsilon)$ error. This is analogous to the use of ordinary QSP, which also approximates the target channel to within $O(\epsilon)$ error. The crucial advantage of stochastic QSP, however, is its efficiency: it achieves this approximation while requiring asymptotically half the number of queries as ordinary QSP. 

This result holds even in more complex scenarios. For instance, if the target operation is controlled by an ancilla qubit (e.g., in the QSP-based phase estimation algorithm of Ref.~\cite{martyn2021grand}), then stochastic QSP can approximate the controlled target operation by an analogous mixture of controlled QSP operations. This result relies on an extension of the mixing lemma to controlled operations, which we prove in Appendix~\ref{app:Mixing_Lemmas}.

Similarly, stochastic QSP is fully compatible with algorithms that use the entire QSP unitary that block-encodes $P(A)$. We show in Section~\ref{sec:Generalizations} that the ensemble of unitaries used in stochastic QSP approximates not just the block-encoded operator $P(A)$, but also the entire QSP unitary operator $U_{\vec{\phi}}[A]$ that encodes $P(A)$. Examples of such algorithms include the ground state energy estimation and ground state projection algorithms considered in Ref.~\cite{lin2020optimal}, as well as the use of amplitude amplification where the projector is constructed through a polynomial approximation $\Pi = P(A)$~\cite{Gily_n_2019}.

\section{Generalizations}\label{sec:Generalizations}
In the previous section, we introduced stochastic QSP tailored specifically to polynomial approximations obtained from truncated Chebyshev expansions. However, it turns out that this specialization is not necessary, and that stochastic QSP is readily generalized to a much broader range of use cases. In this section, we show how our method applies to both definite- and indefinite-parity polynomials, Taylor series, trigonometric polynomials, generalized QSP \cite{motlagh2024generalized}, and approximations
of the entire QSP unitary.

\subsection{Definite and Indefinite Parity}

The statement of Theorem~\ref{thm:StochasticQSP} applies to functions $F(x)$ of indefinite parity\footnote{Recall that a function is said to have definite parity if it is either even or odd. Otherwise it has indefinite parity.} and produces an ensemble of polynomials $\{ P_j \}$ that also have indefinite parity. However, if $F(x)$ has definite parity, it turns out that this construction preserves the parity.

The parity of the implemented polynomials $P_j$ is important because definite-parity polynomials admit simpler implementations via QSP. Indeed, by construction QSP can only produce polynomials of definite parity (see Eqs.~\eqref{eq:QSP_seqeunce} and~\eqref{eq:qsp_conditions}). In general, indefinite parity polynomials require using a linear combination of unitaries circuit, which demands an extra ancilla qubit and rescales the resulting block-encoding by 1/2. This rescaling can be undesirable in algorithm construction because it may necessitate amplitude amplification to be corrected, and can be avoided when the target function $F(x)$ has definite parity in the first place. We would like to retain this optimized performance in the context of stochastic QSP, which we can show is indeed the case.

\begin{corollary} 
Consider the setting of Theorem~\ref{thm:StochasticQSP}, but also suppose that $F(x)$ has definite party. Then there exists an ensemble of polynomials $\{P_j\}$ with the same parity that satisfy the conditions of the theorem.
\end{corollary}
\begin{proof} If $F(x) = \sum_{n=0}^\infty c_n T_n(x)$ has even (odd) parity then all $c_n$ for odd (even) $n$ must vanish. In other words, the function is only supported on even (odd) Chebyshev polynomials. Observe that all $P_j$ in the construction are supported only on subsets of the support of $F(x)$. Hence they must also be even (odd).
\end{proof}

Note that this corollary also extends to complex-valued functions, which are usually approximated by QSP by decomposition into their real and imaginary components. Therefore, for an arbitrary target function, realizing stochastic QSP requires a circuit no more complicated than a QSP circuit that approximates the target function.

\subsection{Taylor Series}\label{sec:StochasticQSPArbBases}

As we discussed in Sec.~\ref{sec:PolyApproxSmooth}, Chebyshev expansions of a smooth $C^\infty$ functions admit exponentially-decaying coefficients, and thus yield polynomial approximations that meet the requisite conditions for stochastic QSP. Functions of interest in the quantum algorithms literature commonly exhibit this smoothness property (like $\cos(x)$ or $\exp(-\beta x)$), or are well-approximated by functions that do (like how the step function is approximated by $\text{erf}(kx)$). 

As such, we often desire a closed-form expression for the coefficients in the Chebyshev expansion, allowing us to give concrete guarantees for the values of $C$ and $q$ required by Theorem~\ref{thm:StochasticQSP}. However, for certain functions like $\sqrt{x}$ and $-x\ln(x)$, which are only smooth in certain domains, obtaining a closed form expression for the Chebyshev coefficients is cumbersome, and the literature generally works with Taylor series instead (see for example Theorem~68 of \cite{Gily_n_2019} and its applications). Fortunately, stochastic QSP directly generalizes to Taylor series, and expansions into bases of bounded polynomials more generally. 

\begin{corollary} \label{cor:basisfuncs} Suppose $\{ B_n (x) \} $ are a collection of basis functions of degree $n$ respectively, which are all bounded as $\|B_n\|_{[-1,1]} \leq 1$. Then the statement of Theorem~\ref{thm:StochasticQSP} holds with $B_n(x)$ in place of $T_n(x)$.
\end{corollary}
\begin{proof} This follows from the fact that the only property of the Chebyshev polynomials $T_n(x)$ leveraged in the proof of Theorem~\ref{thm:StochasticQSP} is that they were bounded as $\| T_n\|_{[-1,1]} = 1$.
\end{proof}

Taylor series methods derive their accuracy from the analysis in Corollary~66 of Ref.~\cite{Gily_n_2019}. The basic idea is the following. Suppose $F(x)$ is analytic and bounded on the interval $[-1,1]$, and our goal is to approximate it by a polynomial on the interval $[-1+\delta, 1-\delta]$ for small $\delta < 1$. If $F(x) = \sum_{n=0}^\infty c_n x^n$ is the Taylor series of $F(x)$ with coefficients $|c_n| \leq 1$, then the error from truncating at degree $d$ is:
\begin{align}
\sup_{|x| \leq 1-\delta}  \left|\sum_{n=d+1}^\infty c_n x^n \right| \leq \sum_{n=d+1}^\infty  (1-\delta)^n.
\end{align}
We immediately obtain an exponential decay in truncation error. 

This is a slightly weaker condition than the one required for Corollary~\ref{cor:basisfuncs}; we require an exponential bound on the coefficients themselves rather than on the truncation error. But if we are willing to approximate the stretched function $F( (1-\delta) x )$ on the interval $[-1,1]$ instead, then substituting the Taylor expansion yields
\begin{align}
F((1-\delta)x) = \sum_{n=0}^\infty c_n (1-\delta)^n x^n 
\end{align}
where the new coefficients $c_n (1-\delta)^n $ exhibit the desired exponential decay, thus rendering this Taylor series expansion amenable to parallel QSP.

\subsection{Trigonometric Polynomials}
Recent works \cite{Dong_2022, Katabarwa2023EarlyFQ} have considered a model of quantum computation in which a constant-size control register is strongly coupled to many qubits with an otherwise local connectivity graph. In such an architecture, controlled time evolution can be implemented through the Trotter approximation, but Hamiltonian simulation via QSP techniques remains out of reach due to the small size of the control register. Nonetheless QSP-like transformations of a Hamiltonian $H$ can be implemented through applying QSP to a controlled time evolution operator.

A controlled time evolution operator is a block-encoding of $e^{i H t}$. Applying QSP to this block-encoding generates trigonometric polynomials $\sum_n c_n e^{i n H t}$. If we select $t = \pi/\|H\|$ and let our variable of interest be $\theta=Ht = \pi H/\|H\|$, then we can approximate functions $F : [-\pi, \pi] \to [-1,1]$ using degree-$d$ trigonometric polynomials:
\begin{align}
    F\left(\theta\right) &\approx \sum_n c_n e^{i n \theta}.
\end{align}
Hence, by selecting $B_n(\theta) =e^{in\theta}$ for the basis functions, we see how Corollary~\ref{cor:basisfuncs} applies in this setting as well. Any method for constructing trigonometric polynomials can be used as long as the coefficients $c_n$ decay exponentially. Conveniently, we show in Appendix~\ref{app:SmoothFunctions} that for $C^\infty$ functions $F(\theta)$, their Fourier series have exponentially decaying coefficients. We see that, due to the relationship between Fourier expansions and Chebyshev expansions, stochastic QSP applies equally well to trigonometric polynomials as to regular polynomials.

\subsection{Generalized QSP}

Recently a technique was proposed \cite{motlagh2024generalized} for optimizing QSP implementation, specifically when the block-encoded operator $U$ is unitary and is encoded via the controlled-$U$ operation. In this situation the usual constraints on parity can be lifted when synthesizing complex polynomials, enabling polynomials of indefinite parity to be generated directly through QSP and avoiding rescaling from using LCU circuit. For real polynomials it is possible to remove parity constraints using Theorem~56 of \cite{Gily_n_2019}, but this introduces an undesirable factor of $1/2$ as discussed earlier. Using the methods of \cite{motlagh2024generalized} this factor can be avoided.

By this reasoning, stochastic QSP is also compatible with the construction of Ref.~\cite{Berry_2024_Doubling}, which halves the asymptotic cost of QSP-based Hamiltonian simulation by using generalized QSP. This is achieved by designing a block encoding of both the quantum walk operator and its inverse, to which generalized QSP may be applied to approximate $e^{-iHt}$ at roughly half the cost of ordinary QSP. Stochastic QSP could be applied on top of this algorithm to further reduce the cost by another factor of $1/2$, equating to a total cost reduction of $1/4$.

\subsection{Approximating the Entire QSP Unitary}

QSP methods typically project out a single polynomial from the QSP sequence, conventionally taken to be the polynomial $P(A)$ encoded in the $|0\rangle\langle 0|$ block (see Eqs.~\eqref{eq:QSP_seqeunce} and~\eqref{eq:QET_QSVT_seq}). Indeed stochastic QSP is concerned with approximating evolution under a channel by projecting out this component. However, in some situations \cite{Lin_2020, Rall_2021} and in amplitude amplification, we are also interested in the other elements of the QSP sequence---particularly those that influence the measurements of the ancilla qubit(s) used to construct the block-encoding. 

These more complex uses of QSP rely on the entire QSP unitary $U_{\vec{\phi}}[A]$ that block-encodes the transformed operator $P(A)$. For simplicity, let us denote this unitary that block encodes $P(A)$ by $U[P(A)]$; similarly, let $U[F(A)]$ denote the unitary that block encodes the target function $F(A)$. To show that stochastic QSP can be applied in these situations as well, we demonstrate that stochastic QSP also approximates the operator $U[P(A)]$ as a whole, rather than just the $P(A)$ sub-block.

\begin{corollary} \label{cor:qsp_unitary} 
In the setting of Theorem~\ref{thm:StochasticQSP}, consider the ensemble of quantum circuits that implement stochastic QSP (i.e. the QSP circuits that implement $\{P_j(A)\}$), upon leaving the QSP ancilla qubit(s) unmeasured. Denoting this ensemble by $\{U_j\}$, the quantum channel $\rho \to \sum_j p_j U_j \rho U_j^\dag $ approximates the map $\rho \to U[F(A)]\rho U[F(A)]^\dag$ to error $O(\epsilon)$ in diamond norm.
\end{corollary}

In the sake of brevity, we defer this proof to Appendix~\ref{app:proof_qsp_unitary}. We also note that this result naturally extends to QSVT, where the desired mapping instead acts on the singular values and singular vectors.

\section{Applications}\label{Sec:Applications}

\begin{figure*}
    \includegraphics[width=.99\linewidth]{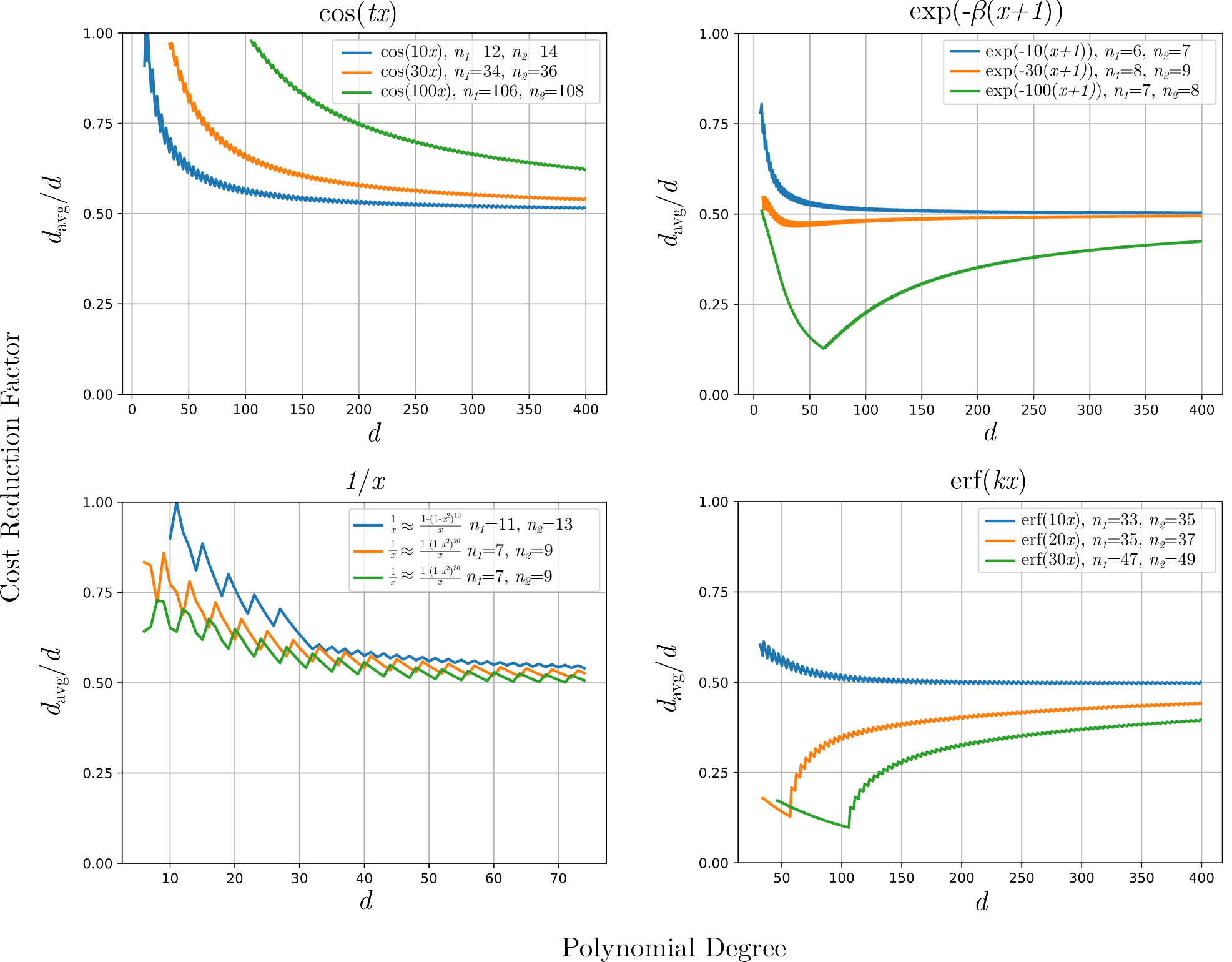}
    \caption{\label{fig:plots} Average reduction in polynomial degree and hence circuit depth for some common polynomials in the literature. $d_\text{avg}$ is computed directly from the probabilities $p_i$ in Eq.~\eqref{eqn:probabilities}, \emph{not} from the bound in Eq.~\eqref{eqn:davgratio}. The cutoff degree $d^*$ is selected according to Eq.~\eqref{eqn:dstarselect}, which depends on a choice of $C$ and $q$ obtained from $n_1,n_2$. See the explanation in the main text. }
\end{figure*}

Here we apply stochastic QSP to several common polynomials used in the quantum algorithms literature to assess its performance in practice. Our focus will be on polynomial approximations of smooth functions, such that Theorem~\ref{thm:StochasticQSP} is applicable. 

As we mentioned in Sec.~\ref{sec:PolyApproxSmooth}, this includes polynomials that approximate functions with a finite number of singularities. This is achieved by excluding regions around the singularities, allowing the function to be replaced by a smooth interpolation that can then be approximated by a polynomial. As we show, this renders stochastic QSP applicable to commonly-used discontinuous functions, like the step function $\Theta(x)$ and the inverse function $1/x$. 

We study the following four polynomials:

\begin{enumerate}
    \item The Jacobi-Anger expansion of cosine \cite[Lemma 57]{Gily_n_2019}:
    \begin{align}
    \cos(tx) = J_{0}(t) + 2\sum_{n=1}^\infty (-1)^n J_{2n}(t) T_{2n}(x),
    \end{align}
    where the $J_{2n}(t)$ are the Bessel functions of the first kind. To achieve additive error at most $\epsilon$, this series may be truncated at degree $d = O\big(|t| + \frac{\log(1/\epsilon)}{ \log(e + \log(1/\epsilon)/|t|)} \big)$. This polynomial, in conjunction with the analogous expansion for $\sin(tx)$, furnishes an algorithm for Hamiltonian simulation with near-optimal query complexity~\cite{Low_2017}. 
    
    \item The Jacobi-Anger expansion of an exponential decay \cite[Lemma~14]{low2017hamiltonian}:
    \begin{align}
    e^{-\beta(x+1)} = e^{-\beta }\left[ I_0(\beta) + 2\sum_{n=1}^\infty  I_n(\beta) T_b(-x)\right],
    \end{align}
    where the $I_n(\beta)$ is the modified Bessel functions of the first kind. This expansion may be truncated at degree $d = O\big(\sqrt{(\beta + \log(1/\epsilon))\log(1/\epsilon)}\big) = O(\sqrt{\beta} \log(1/\epsilon)) $ to achieve error at most $\epsilon$. Naturally, the resulting polynomial is commonly used for imaginary time evolution~\cite{Silva_2023, chan2023simulating}, thermal state preparation~\cite{Gily_n_2019}, and partition function estimation~\cite{tosta2023randomized}. 
    
    \item A smooth approximation of $1/x$ in a domain away from the origin \cite[Lemma~14]{Childs_2017}:
    \begin{equation}
    \begin{aligned}
    \frac{1}{x} &\approx \frac{1- (1-x^2)^b}{x} \\
    &= 4\sum_{n=0}^{b-1} (-1)^n 2^{-2b} \sum_{m=n+1}^b \binom{2b}{b+m} T_{2n+1}(x)
    \end{aligned}
    \end{equation}
    for an even integer $b \gg 1$. While this series is a degree $O(b)$ polynomial, its coefficients decay rapidly, such that it can be further truncated to degree $d= O\big( \sqrt{b \log(b/\epsilon)} \big)$ while guaranteeing error at most $\epsilon$. The resulting polynomial is particularly useful for inverting matrices and thus solving linear systems of equations. If we take the (non-zero) eigenvalues of the matrix to be lower-bounded as $|\lambda| \geq \lambda_{\text{min}}$, then in order to ensure that the polynomial approximation behaves as $\approx 1/x$ over the range of eigenvalues, it suffices to choose chooses $b = O \big( \frac{1}{\lambda_{\text{min}}^2} \log(1/\lambda_{\text{min}} \epsilon) \big)$. This corresponds to a truncation degree $d = O(\kappa \log(\kappa/\epsilon))$, where $\kappa := 1/\lambda_{\text{min}}$ is the condition number. For completeness, we note that algorithms with improved performance have very recently been discovered \cite{dalzell2024shortcut}.

    \item An approximation of $\erf(kx)$ obtained from integrating the Jacobi-Anger expansion of a Gaussian \cite[Corollary~4]{low2017hamiltonian}:
    \begin{equation}
    \begin{aligned}
    &\erf(kx) = \frac{2ke^{-\frac{k^2}{2}}}{\sqrt{\pi}} \Bigg[I_0\left(\frac{k^2}{2}\right)x \\
    &\qquad+ \sum_{n=1}^\infty I_0\left(\frac{k^2}{2}\right) (-1)^n \left(  \frac{T_{2n+1}(x)}{2n+1} - \frac{T_{2n-1}}{2n-1}\right) \Bigg].
    \end{aligned}
    \end{equation}
    To achieve error $\epsilon$, it suffices to truncate this series at degree $d = O(\sqrt{(k^2 + \log(1/\epsilon))\log(1/\epsilon)} \big) = O(k \log(1/\epsilon))$. In practice, this polynomial is used to approximate the step function $\Theta(x)$ by selecting $k \gg 1$. Notable applications of this approximation include the search problem~~\cite{martyn2021grand}, phase estimation~\cite{martyn2021grand, Rall_2021} and ground state preparation~\cite{Lin_2020, Dong_2022}.
\end{enumerate}

All four of these functions feature a \emph{cost parameter}, namely $t,\beta,b,$ and $k$ respectively, whose value determines the truncation degree necessary to achieve an accurate approximation.

We apply stochastic QSP to these polynomials, and illustrate the cost reduction ratio $d_\text{avg}/d$ as a function of $d$ in Figure~\ref{fig:plots}. We rely on the following procedure to compute $d_\text{avg}$. First, we select an integer $n_1$, and then numerically determine values of $C,q$ such that $c_n \leq Ce^{-qn} $ holds for $n \geq n_1$. This is achieved by selecting another integer $n_2 > n_1$ and computing $C,q$ such that $Ce^{-qn}$ goes through the points $(n_1, c_{n_1})$ and $(n_2, c_{n_2})$. In doing so, we choose $n_1,n_2$ to guarantee $c_n \leq Ce^{-qn} $ indeed holds for all $n \geq n_1$, and also to minimize $\log(C)/q$ so as to heuristically reduce the dominant term in the bound on $d_\text{avg}$ in Eq.~\eqref{eqn:davgratio}. We also select $n_1,n_2$ independent of the degree $d$; we find that $n_1$ naturally converges to the degree at which $|c_n|$ starts to decay exponentially. For $\cos(tx)$ and $\erf(kx)$ we find that this regime sets in later as the respective cost parameter increases.
Lastly, for each $d$, the cutoff degree $d^*$ is computed from Eq.~\eqref{eqn:dstarselect}. Then $d_\text{avg}$ is obtained by explicitly computing the probabilities $p_j$ from Eq.~\eqref{eqn:probabilities} and calculating the corresponding weighted average of degrees. 

In Figure~\ref{fig:plots} we observe the desired phenomenon: the cost reduction ratio $d_\text{avg}/d$ approaches $1/2$ as $d$ increases, with a discrepancy scaling as $O(1/d)$. Additionally, as the cost parameter increases, the magnitude of the error terms $\sim \log(C)/qd$ can also increase, resulting in a later approach to $1/2$. This is expected because a larger cost parameter requires a larger degree to maintain the same level of approximation. These results indicate that indeed, stochastic QSP reduces the query complexity of QSP by approximately 1/2 in regimes of practical interest. 

A more surprising phenomenon is that for some functions and values of the cost parameter, $d_\text{avg}/d$ approaches $1/2$ \emph{from below}, resulting in improved performance for small $d$. This arises because we determine $d_\text{avg}$ by choosing values $C,q$ so as to minimize the deviation $\log(C)/q$. In some cases, this can result in a value $C<1$, or equivalently $\log(C)<0$, which causes the ratio $d_\text{avg}/d$ to deviate from $1/2$ by a negative value as per Eq.~\eqref{eqn:davgratio}. Moreover, the sawtooth behavior observed for each function in Figure~\ref{fig:plots} is explained by the presence of the ceiling function in the definition of $d^*$ in Eq.~\eqref{eqn:dstarselect}. When $d$ is small then the effect of this rounding is more pronounced.

Lastly, we emphasize that the cost reduction realized by stochastic QSP hinges on the $\log(1/\epsilon)$ term in the degree of the polynomial approximation. In scenarios where the degree scales multiplicatively with $\log(1/\epsilon)$ (e.g., approximations to $e^{-\beta x}$, $\text{erf}(kx)$, and $1/x$), the cost reduction quickly approaches 1/2. Conversely, in cases where the degree scales additively in $\log(1/\epsilon)$ (e.g., approximations to $\cos(tx)$ and $\sin(tx)$), the cost reduction only approaches 1/2 when the $\log(1/\epsilon)$ term dominates. For instance, in Hamiltonian simulation, the $\log(1/\epsilon)$ term may not dominate in the large $t$ limit, and consequently the practical cost reduction will instead lie between 1 and 1/2.

\section{Conclusion}
By merging quantum signal processing and randomized compiling, we have developed stochastic QSP. As we showed, stochastic QSP enables us to lift a QSP algorithm to its randomized counterpart, and simultaneously shrink the circuit complexity by a factor of 2. We empirically verified this cost reduction for various algorithms of interest, including real/imaginary time evolution, matrix inversion, phase estimation, and ground state preparation. This reduction can also be interpreted as enabling a cost parameter in the underlying function to increase (e.g. a longer time $t$ in Hamiltonian simulation) without changing the query complexity. In aggregate, this result demonstrates that classical randomness is a useful resource for quantum computing, and can help bring quantum algorithms closer to the near-term.

Moreover, in this work we did not consider noisy gates, but rather assumed the ability to perform QSP exactly. As such, we leveraged randomized compiling to suppress error in QSP polynomial approximations to functions. Nonetheless, as randomized compiling can also suppress noise in erroneous gates~\cite{Wallman_2016}, this suggests that a practical implementation of stochastic QSP could benefit from also randomizing over the gate set, as a sort of doubly-stochastic channel. Along these lines, it would be interesting to study the requirements and conditions for implementing stochastic QSP on near-term quantum hardware. 

The performance improvement realized through randomized compiling suggests further uses of this technique in quantum information. While here we have applied randomized compiling to quantum algorithms via QSP, it is likely that randomized compiling can confer a similar advantage to the traditional constructions of quantum algorithms (e.g., Grover search~\cite{grover1996fast}, or conventional phase estimation via the quantum Fourier transform~\cite{kitaev1995quantum}). Further, it would be interesting to search for problems for which randomized compiling (or a variant thereof) confers a super-quadratic suppression of error, translating to a cost reduction by a factor smaller than 1/2 in our context. With an eye toward applications, the mixing lemma could also be used to better understand and generate random unitaries and unitary designs~\cite{Gross_2007, Fisher_2023, Harrow_2009_2, Brand_o_2016}, and perhaps even be integrated with randomized measurements~\cite{Elben_2022} to improve near-term protocols for studying quantum systems. Likewise, there is an absence of randomized compiling in classical simulation algorithms, which could admit similar improvements for problems aimed to emulate quantum mechanics. Given the ubiquity of random processes in the physical world, it is only natural that we can harness randomness to gain deeper insights into quantum systems.

\textit{Acknowledgements}: The authors thank Pawel Wocjan, Anirban Chowdhury, Isaac Chuang, and Zane Rossi for useful discussion and feedback. JMM acknowledges the gracious support of IBM Quantum through the IBM internship program.

\textit{Note}: After posting this work, we were made aware of a recent paper~\cite{wang2024faster} that also mixes over polynomials to suppress error. Despite this similarity, our work makes novel contributions: we provide (1) an analytic construction of the ensemble of polynomials and associated probability distribution, and (2) a rigorous argument for the cost reduction by a factor of $1/2$. These key contributions establish stochastic QSP as a versatile framework with strong performance guarantees.

\bibliography{References}

\appendix
\section{Polynomial Approximation to Smooth Functions}\label{app:SmoothFunctions}
Here we recall theorems on polynomial approximations to smooth functions, specifically the decay of Fourier coefficients and Chebyshev coefficients of smooth functions. For more details on these results, see Refs.~\cite{katznelson2004introduction, boyd2001chebyshev}.

\subsection{Decay of Fourier Coefficients}
Fourier analysis can be used to determine the rate of decay of the Fourier coefficients of smooth functions. The following is a classic result~\cite{katznelson2004introduction}, for which we provide a brief proof:
\begin{theorem}\label{thm:Fourier_Coef_Decay}
    Let $G(\theta)$ be a function with period $2\pi$ and Fourier series 
    \begin{equation}
        G(\theta) = \frac{a_0}{2} + \sum_{n=-\infty}^\infty a_n e^{in\theta}.
    \end{equation}
    If $G(\theta)$ is a $C^k$ function on $\theta \in [0,2\pi)$ (i.e., continuous and differentiable through order $k$), then the Fourier coefficients decay polynomially as 
    \begin{equation}
        a_n = o\Big(\frac{1}{n^k}\Big).
    \end{equation}
    Similarly, if $G(\theta)$ is a $C^\infty$ function on $\theta \in [0,2\pi)$, then the Fourier coefficients decay super-polynomially as 
    \begin{equation}
        a_n = O\big(e^{-qn^r} \big)
    \end{equation}
    for some $q,r > 0$. 
\end{theorem}
\begin{proof}
    The Fourier coefficients are
    \begin{equation}\label{eq:Fourier_Coef}
        a_n = \frac{1}{2\pi} \int_0^{2\pi} G(\theta) e^{-i n\theta} d\theta . 
    \end{equation}
    Using integration by parts repeatedly, we can write this as 
    \begin{equation}
        \begin{aligned}
            a_n &= \frac{1}{2\pi} \int_0^{2\pi} G(\theta) e^{-i n\theta} d\theta \\
            &= \frac{i}{n} \frac{1}{2\pi} \int_0^{2\pi} G'(\theta) e^{-in\theta} d\theta \\
            & = \ \ ... \\ 
            &= \frac{i^k}{n^k} \frac{1}{2\pi} \int_0^{2\pi} G^{(k)}(\theta) e^{-in\theta} d\theta . 
        \end{aligned}
    \end{equation}
    This implies that the Fourier coefficients of the $k$th derivative $G^{(k)}(\theta)$ are $(-in)^k a_n$. 

    According to the Riemann-Lebesgue lemma~\cite{katznelson2004introduction}, the Fourier coefficients of a smooth function go to 0 as $n\rightarrow \infty$. Therefore, if $G(\theta)$ is a $C^k$ function, then $G^{(k)}(\theta)$ is continuous, and thus its Fourier coefficients decay as $\lim_{n\rightarrow \infty} (-in)^k a_n = 0$. This limit implies that $a_n = o(1/n^k)$. 
    
    On the other hand, if $G(\theta)$ is a $C^\infty$ function, then by an analogous argument, $\lim_{n\rightarrow \infty} n^k a_n = 0 $ for all integers $k \geq 1$. This implies that $a_n$ decays super-polynomially, i.e. as $a_n = O\big(e^{-qn^r} \big)$ for some $q,r>0$. 
\end{proof}

\subsection{Decay of Chebyshev Coefficients}
An analogous result can be derived for the decay of Chebyshev coefficients. Recall that the $n$th Chebyshev polynomial $T_n(x)$ is a degree $n$ polynomial defined on $x\in [-1,1]$ as
\begin{equation}
    T_n(x) = \cos(n \arccos(x)).
\end{equation}
It can be shown that $T_n(x)$ is a degree $n$ polynomial of definite parity (either even or odd, depending on $n$) and bounded magnitude $|T_n(x)|_{[-1,1]} = 1$~\cite{boyd2001chebyshev}. The Chebyshev polynomials provide a convenient basis for expanding functions on $x\in [-1,1]$. A function $F(x)$ can be expanded as
\begin{equation}
    F(x) = \frac{c_0}{2} + \sum_{n=1}^\infty c_n T_n(x) , 
\end{equation}
where
\begin{equation}
     c_n = \frac{2}{\pi} \int_{-1}^1 \frac{F(x) T_n(x)}{\sqrt{1-x^2}} dx 
\end{equation}
are the Chebyshev coefficients for all $n\geq 0$. 

By the relation between Chebyshev series and Fourier series, it can be shown that the Chebyshev coefficients decay in a manner analogous to the Fourier coefficients: 
\begin{theorem}\label{thm:Cheby_Coef_Decay}
    Let $F(x)$ be a function on $x\in [-1,1]$ with the Chebyshev series 
    \begin{equation}
        F(x) = \frac{c_0}{2} + \sum_{n=1}^\infty c_n T_n(x) , 
    \end{equation}
    If $F(x)$ is a $C^k$ function on $x \in [-1,1]$, then the Chebyshev coefficients decay polynomially as 
    \begin{equation}
        c_n = o\Big(\frac{1}{n^k}\Big).
    \end{equation}
    Similarly, if $F(x)$ is a $C^\infty$ function on $x \in [-1,1]$, then the Chebyshev coefficients decay super-polynomially as 
    \begin{equation}
        c_n = O\big(e^{-qn^r} \big)
    \end{equation}
    for some $q,r > 0$. 
\end{theorem}
\begin{proof}
    Let $\theta = \arccos(x)$, or equivalently $x = \cos\theta$. Using this change of variables, we can re-express the Chebyshev coefficients as
    \begin{equation}
    \begin{aligned}
     c_n &= \frac{2}{\pi} \int_{-1}^1 \frac{F(x) T_n(x)}{\sqrt{1-x^2}} dx = \frac{2}{\pi} \int_0^{\pi} F(\cos(\theta)) \cos(n \theta) d\theta \\
     &= \frac{1}{2\pi} \int_0^{2\pi} F(\cos(\theta)) e^{-i n \theta} d\theta 
    \end{aligned}
    \end{equation}
    where the last equality stems from $F(\cos(\theta))$ being an even function of $\theta$. Comparing with Eq.~\eqref{eq:Fourier_Coef}, we see that the Chebyshev coefficients $c_n$ are equal to the Fourier coefficients of the function $G(\theta) := F(\cos(\theta)) $. Therefore, if we can show that $F(x)$ being a $C^k$ function of $x$ implies that $G(\theta) = F(\cos(\theta))$ is a $C^k$ function of $\theta$, then we can inherit the result of Theorem~\ref{thm:Fourier_Coef_Decay} to prove the purported decay pattern of the coefficients $c_n$. 

    This is easy to show. Observe that the first derivative of $G(\theta)$ is
    \begin{equation}
        G'(\theta) = F'(\cos\theta) (-\sin \theta) = - F'(x) \sin \theta .
    \end{equation}
    If $F(x)$ is a $C^1$ function of $x=\cos\theta$, such that $F'(x)$ is a continuous function of $x$, then $G'(\theta) = -F'(\cos\theta) \sin (\theta)$ is a composition of continuous functions of $\theta$. Therefore, $G'(\theta)$ is is also a continuous function of $\theta$, which implies that $G(\theta)$ is a $C^1$ function of $\theta$. A similar trend persists at higher orders: by using the chain rule and product rule, the $k$th derivative $G^{(k)}(\theta)$ can be expressed as a linear combination of products of derivatives $F^{(j)}(\cos \theta)$ for $j\leq k$ and trigonometric polynomials $\sin^a(\theta) \cos^b(\theta)$ for some integers $a,b\geq 0$. Thus, if $F(x)$ is a $C^k$ function of $x$, then this expression for $g^{(k)}(\theta)$ is also continuous, which implies that $G(\theta)$ is a $C^k$ function. 

    Therefore, by appealing to Theorem~\ref{thm:Fourier_Coef_Decay}, we see that if $F(x)$ is a $C^k$ function, then the Chebyshev coefficients $c_n$ decays polynomially as $c_n = o(1/n^k)$. Similarly, if $F(x)$ is a $C^\infty$ function, then $c_n$ decays super-polynomially as $c_n = O\big(e^{-qn^r}\big)$ for some $q,r > 0$.  
\end{proof}

\section{Extensions of the Mixing Lemma}\label{app:Mixing_Lemmas}
Here we present generalizations of the mixing lemma to non-unitary evolution and block-encodings. The proofs of these variants parallel the proof of the mixing lemma presented in Ref.~\cite{Campbell_2017}. 

\subsection{The Generalized Mixing Lemma}

Consider a mixing lemma for arbitrary operators, including non-unitary evolution. In this scenario, we we wish to implement a channel $\mathcal{S}(\rho) = S \rho S^\dag$, where $S$ is not necessarily unitary, yet we only have access to operators $R_j$ that approximate $S$. Then we can prove the following:
\begin{lemma}[Generalized Mixing Lemma]\label{thm:NonUnitaryMixing_app}
    Let $S$ be a target operator, possibly non-unitary, and $\mathcal{S}(\rho) = S \rho S^\dag $ the corresponding channel. Suppose there exist operators $R_j$ and a probability distribution $p_j$ that approximate $S$ as
    \begin{equation}
    \begin{aligned}
        & \| R_j - S \| \leq a \text{ for all } j,\\
        & \Big\| \sum_j p_j R_j - S \Big\| \leq b ,
    \end{aligned}
    \end{equation}
    for some $a,b > 0$. Then, the corresponding channel $\Lambda(\rho) = \sum_j p_j R_j \rho R_j^\dag$ approximates the channel $\mathcal{S}$ as
    \begin{equation}
        \| \Lambda - \mathcal{S} \|_\diamond \leq a^2 + 2b \|S \|. 
    \end{equation}
\end{lemma}
\begin{proof}
Paralleling the proof of the mixing lemma in Ref.~\cite{Campbell_2017}, let $\delta_j = R_j - S$. This obeys $\| \delta_j \| \leq a$ and $\big\| \sum_j p_j \delta_j \big\| \leq b$ by the assumed conditions of the theorem. The action of the channel $\mathcal{S}$ can then be expanded as
\begin{equation}
\begin{aligned}
    \Lambda(\rho) &= \sum_j p_j R_j \rho R_j^\dag \\
    &= \sum_j p_j (S+\delta_j) \rho (S+\delta_j)^\dag \\
    &= S\rho S^\dag + \Big(\sum_j p_j \delta_j \Big)\rho S^\dag \\
    & \quad + S \rho \Big(\sum_j p_j \delta_j^\dag \Big) + \sum_j \delta_j \rho \delta_j^\dag.
\end{aligned}
\end{equation}
We can then bound the 1-norm $\|\Lambda(\rho) - \mathcal{S}(\rho) \|_1$ via the triangle inequality as 
\begin{equation}\label{eq:random_channel_1norm}
\begin{aligned}
    \| \Lambda(\rho) - \mathcal{S}(\rho) \|_1 \leq & \ \Big\| \Big(\sum_j p_j \delta_j \Big)\rho S^\dag \Big\|_1 \\
    & + \Big\| S \rho \Big(\sum_j p_j \delta_j^\dag \Big) \Big\|_1 + \Big\| \sum_j \delta_j \rho \delta_j^\dag \Big\|_1 .
\end{aligned}
\end{equation}

The first term on the right hand side of Eq.~\eqref{eq:random_channel_1norm} can be bounded by appealing to Holder's inequality (specifically, $\| AB \|_1 \leq \|A\| \|B\|_1$ and $\| AB \|_1 \leq \|A\|_1 \|B\|$ in this context):
\begin{equation}
\begin{aligned}
    \Big\| \Big( \sum_j p_j \delta_j \Big) \rho S^\dag \Big\|_1 &\leq 
    \Big\| \sum_j p_j \delta_j \rho \Big\|_1 \| S^\dag \| \\
    &\leq \Big\| \sum_j p_j \delta_j \Big\| \| \rho \|_1 \| S \| \\
    &\leq \| S \| b, 
\end{aligned}
\end{equation}
where we have used that $\| \rho \|_1 = 1 $. Analogously, the second term in Eq.~\eqref{eq:random_channel_1norm} can be upper bounded by $\| S \| b$. Again invoking Holder's inequality, the last term in Eq.~\eqref{eq:random_channel_1norm} can be bounded as 
\begin{equation}
\begin{aligned}
    \Big\| \sum_j \delta_j \rho \delta_j^\dag \Big\|_1 &\leq \sum_j p_j \|\delta_j \rho \delta_j^\dag\|_1 \\
    &\leq \sum_j p_j \| \delta_j \rho \|_1 \| \delta_j^\dag \| \\
    &\leq \sum_j p_j \| \delta_j \| \|\rho\|_1 \| \delta_j^\dag \| \\
    & \leq \sum_j p_j \| \delta_j \| \| \delta_j \| \\
    & \leq a^2. 
\end{aligned}
\end{equation}
Therefore we have
\begin{equation}\label{eq:NonUnitary_Mixing_1norm}
    \big\| \Lambda(\rho) - \mathcal{S}(\rho) \big\|_1 \leq a^2 + 2b \|S \|, 
\end{equation}
for all density matrices $\rho$. This translates to a diamond norm bound:
\begin{equation}
    \begin{aligned}
        \big\| \Lambda - \mathcal{S} \big\|_\diamond &= \sup_{\| \sigma\|_1 \leq 1} \| (\Lambda \otimes \mathcal{I} )(\sigma) - (\mathcal{S} \otimes \mathcal{I})(\sigma) \|_1 \\
        & \leq a^2 + 2b \|S \| ,
    \end{aligned}
\end{equation}
because appending the identity channel $\mathcal{I}$ does not alter the proof of 1-norm bound in Eq.~\eqref{eq:NonUnitary_Mixing_1norm}, which holds true for all density matrices.
\end{proof}

Notably, this result is analogous to the original mixing lemma, but modified by the spectral norm of $S$. For unitary evolution where $\| S \| = 1$, the generalized mixing lemma reduces to the usual mixing lemma.

\subsection{The Mixing Lemma for Block-Encodings}
We can consider an analog of the mixing lemma for block-encodings. In this scenario, we wish to evolve under an operator $S$ (which may be non-unitary) encoded in a unitary $V$, yet we only have access to operators $R_j$ block-encoded in unitaries $U_j$. We first study the case in which the operators are block-encoded in the $|0\rangle \langle 0|$ block of their respective unitaries, which is the simplest form of a block-encoding and most often considered in the literature. Afterwards, we proceed to the general case, in which the operators are encoded through arbitrary projectors $\Pi$ and $\Pi'$.

First consider encodings in the $|0\rangle \langle 0| $ block:
\begin{lemma*}[Mixing Lemma for Block-Encodings: $|0\rangle \langle 0|$ Block]
    Let $V$ be a unitary that block-encodes a (possibly non-unitary) target operator $S$ as $S = (\langle 0 | \otimes I) V ( |0 \rangle \otimes I) $. Suppose there exist unitaries $U_j$ that block-encode operators $R_j$ as $R_j = (\langle 0 | \otimes I)  U_j ( | 0 \rangle \otimes I)$. Also suppose there exists a probability distribution $p_j$ such that
    \begin{equation}
    \begin{aligned}
        &\| R_j  - S \| \leq a \text{ for all } j , \\ 
        &\Big\| \sum_j p_j R_j - S \Big\| \leq b ,
    \end{aligned}
    \end{equation}
    for some $a,b \geq 0$. Then, the corresponding unitary channel $\Lambda(\rho) = \sum_j p_j U_j \rho U_j^\dag$ approximates the action of the channel $\mathcal{V}$ as
    \begin{equation}
        \big\| \bar{\Lambda} - \bar{\mathcal{V}} \big\|_\diamond \leq a^2 + 2b , 
    \end{equation}
    where $\bar{\Lambda}$ and $\bar{\mathcal{V}}$ are channels that access the block-encodings of $\Lambda$ and $\mathcal{V}$ by appending an ancilla qubit and post-selecting:
    \begin{equation}
        \begin{aligned}
            & \bar{\Lambda}(\rho) = (\langle 0 | \otimes I) \cdot  \Lambda \big( | 0 \rangle \langle 0 | \otimes \rho \big) \cdot ( |0 \rangle \otimes I) \\
            & \bar{\mathcal{V}}(\rho) = (\langle 0 | \otimes I) \cdot \mathcal{V} \big( | 0 \rangle \langle 0 | \otimes \rho \big) \cdot ( |0 \rangle \otimes I).
        \end{aligned}
    \end{equation}
\end{lemma*}

\begin{proof}   
    First, observe that by linearity 
    \begin{equation}
        \begin{aligned}
            \bar{\Lambda}(\rho) &=  (\langle 0 | \otimes I ) \Lambda\big( | 0 \rangle \langle 0 | \otimes \rho \big) ( |0 \rangle \otimes I) \\
            &= \sum_j p_j R_j \rho R_j^\dag, 
        \end{aligned}
    \end{equation}
    and 
    \begin{equation}
        \begin{aligned}
            \bar{\mathcal{V}}(\rho) &= (\langle 0 | \otimes I ) \mathcal{V}\big( | 0 \rangle \langle 0 | \otimes \rho \big) ( |0 \rangle \otimes I) \\
            & = S \rho S^\dag. 
        \end{aligned}
    \end{equation}
    Therefore, $\bar{\Lambda}$ and $\bar{\mathcal{V}}$ are analogous to the channels $\Lambda$ and $\mathcal{S}$ considered in Theorem~\ref{thm:NonUnitaryMixing_app}. Likewise, $S$ and $R_j$ obey the requisite conditions of Theorem~\ref{thm:NonUnitaryMixing_app} as per the assumptions of this theorem. Accordingly, Theorem~\ref{thm:NonUnitaryMixing_app} implies that
    \begin{equation}
        \begin{aligned}
            & \| \bar{\Lambda} - \bar{\mathcal{V}} \|_\diamond \leq a^2 + 2b \| S \| \leq a^2 + 2b,
        \end{aligned}
    \end{equation}
    where we have used that $\|S\| \leq 1 $ because $S$ is block-encoded in a unitary.
\end{proof}

Next, consider the general case in which the operators are block-encoded by arbitrary projectors $\Pi$ and $\Pi'$. Specifically, in this case, an operator is block-encoded as $S = \Pi V \Pi' $ for projection operators $\Pi, \Pi'$. In practice, this means that $S$ is accessed by first applying $\Pi'$ to project the quantum state into the block of interest, then applying $V$, and finally applying $\Pi'$ to extract the desired block. With this intuition, we can prove the following:
\begin{lemma}[Mixing Lemma for Block-Encodings: General Encoding]\label{lemma:BlockEncodingMixing_app}
    Let $V$ be a unitary, that block-encodes a (possibly non-unitary) target operator $S$ as $S = \Pi V \Pi'$, where $\Pi, \Pi' $ are orthogonal projectors. Suppose there exist unitaries $U_j$ that block-encode operators $R_j$ as $R_j = \Pi U_j \Pi '$. Also suppose there exists a probability distribution $p_j$ such that
    \begin{equation}
    \begin{aligned}
        &\| R_j  - S \| \leq a \text{ for all } j , \\ 
        &\Big\| \sum_j p_j R_j - S \Big\| \leq b .
    \end{aligned}
    \end{equation}
    Then, the corresponding unitary channel $\Lambda(\rho) = \sum_j p_j U_j \rho U_j^\dag$ approximates the action of the channel $\mathcal{V}$ as
    \begin{equation}
        \big\| \bar{\Lambda} - \bar{\mathcal{V}} \big\|_\diamond \leq a^2 + 2b , 
    \end{equation}
    where $\bar{\Lambda}$ and $\bar{\mathcal{V}}$ are channels that access the block-encodings of $\Lambda$ and $\mathcal{V}$ by applying the projectors $\Pi$ and $\Pi'$:
    \begin{equation}
        \begin{aligned}
            & \bar{\Lambda}(\rho) = \Pi \cdot  \Lambda \big( \Pi' \rho \Pi' \big) \cdot \Pi \\
            & \bar{\mathcal{V}}(\rho) = \Pi \cdot  \mathcal{V} \big( \Pi' \rho \Pi' \big) \cdot \Pi.
        \end{aligned}
    \end{equation}
\end{lemma}

\begin{proof}
    Analogous to the previous proof, observe that by linearity
    \begin{equation}
        \begin{aligned}
            & \bar{\Lambda}(\rho) = \Pi \cdot  \Lambda \big( \Pi' \rho \Pi' \big) \cdot \Pi \\
            & = \sum_j p_j \Pi U_j \Pi' \cdot \rho \cdot \Pi' U_j^\dag \Pi \\
            & = \sum_j p_j R_j \rho U_j^\dag,
        \end{aligned}
    \end{equation}
    and 
    \begin{equation}
        \begin{aligned}
            \bar{\mathcal{V}}(\rho) &= \Pi \cdot  \mathcal{V} \big( \Pi' \rho \Pi' \big) \cdot \Pi \\
            & = \Pi V \Pi' \cdot \rho \cdot \Pi' V^\dag \Pi \\
            & = S \rho S^\dag, 
        \end{aligned}
    \end{equation}
    where we have use the fact the the orthogonal projectors are Hermitian: $\Pi^\dag = \Pi, \ \Pi'^\dag = \Pi'$. Therefore, we again see that $\bar{\Lambda}$ and $\bar{\mathcal{V}}$ are analogous to the channels $\Lambda$ and $\mathcal{S}$ considered in Theorem~\ref{thm:NonUnitaryMixing_app}, which implies that
    \begin{equation}
        \begin{aligned}
            & \| \bar{\Lambda} - \bar{\mathcal{V}} \|_\diamond \leq a^2 + 2b \| S \| \leq a^2 + 2b,
        \end{aligned}
    \end{equation}
    where we have again used that $\|S\| \leq 1 $ because $S$ is block-encoded in a unitary.
\end{proof}

\subsection{The Mixing Lemma for Controlled Operations}\label{app:MixingLemmaControlled}
The mixing lemma naturally also extends to controlled operations. That is, if the unitaries $\{U_j \}$ and probability distribution $p_j$ satisfy the conditions of the mixing lemma in approximating a target unitary $V$, then one can approximate a controlled-$V$ operation to the same level of accuracy by a probabilistic mixture of controlled-$U_j$'s. More precisely, letting $\text{c-}U_j$ denote a controlled $U_j$ operation, we have the following:

\begin{lemma}[Mixing Lemma for Controlled Operations]\label{lemma:Mixing_Controlled}
    Let c-$V = |0\rangle \langle 0 | \otimes I + |1\rangle \langle 1 | \otimes V$ be a controlled unitary, and $\text{c-}\mathcal{V}(\rho) = (\text{c-}V) \rho (\text{c-}V)^\dag $ the corresponding channel. Suppose there exist $m$ unitaries $\{ U_j \}_{j=1}^m$ and an associated probability distribution $p_j$ that approximate $V$ as
    \begin{equation}
    \begin{aligned}
        & \| U_j - V \| \leq a \text{ for all } j,\\
        & \Big\| \sum_{j=1}^m p_j U_j - V \Big\| \leq b ,
    \end{aligned}
    \end{equation}
    for some $a,b > 0$. Then, the channel comprising a mixture of controlled unitaries $\text{c-}\Lambda(\rho) = \sum_{j=1}^m p_j (\text{c-}U_j) \rho (\text{c-}U_j)^\dag$ approximates the controlled channel $\text{c-}\mathcal{V}$ as
    \begin{equation}
        \| \text{c-}\Lambda - \text{c-}\mathcal{V} \|_\diamond \leq a^2 + 2b . 
    \end{equation}
\end{lemma}

\begin{proof}
    To prove this, note that 
    \begin{equation}
        \big\| \text{c-}U_j - \text{c-}V \big\| = \big\| |1\rangle \langle 1| \otimes (U_j - V) \big\| = \| U_j - V \| \leq a,
    \end{equation}
    and similarly
    \begin{equation}
    \begin{aligned}
        \Big\| \sum_j p_j (\text{c-}U_j) - \text{c-}V \Big\| &= \Big\| |1\rangle \langle 1| \otimes \Big(\sum_j p_j U_j - V \Big) \Big\| \\
        &= \Big\| \sum_{j} p_j U_j - V \Big\| \leq b . 
    \end{aligned}
    \end{equation}
    Evidently, the controlled unitaries $\text{c-}U_j$ and $\text{c-}V$ satisfy the conditions of the mixing lemma. Therefore, one can straightforwardly apply the mixing lemma to these operations to conclude that $\| \text{c-}\Lambda - \text{c-}\mathcal{V} \|_\diamond \leq a^2 + 2b$. 
\end{proof}

In the context of stochastic QSP, this lemma establishes that stochastic QSP is compatible with controlled target operations, which are integral to many quantum algorithms (e.g., the QSP-based phase estimation algorithm in Ref.~\cite{martyn2021grand}). Specifically, this result shows that a controlled target operation can be approximated through stochastic QSP as a probabilistic mixture of controlled QSP operations.

\section{Proof of Corollary~\ref{cor:qsp_unitary}}\label{app:proof_qsp_unitary}
For a given target function $F(x)$ approximated by a truncated polynomial $P(x)$, QSP yields a unitary $U[P(A)]$ that block-encodes $P(A)$, and thus approximates $U[F(A)]$ and $F(A)$ respectively. While most quantum algorithms consider only $P(A)$, some algorithms use the entire unitary $U[P(A)]$. For example, QSP-based methods for ground-state preparation and energy estimation \cite{Dong_2022} consider a Hamiltonian $H$, find a polynomial transformation $P(H)$ that approximates a projector onto an energy range $\Pi_{\leq E} \approx P(H)$, and construct its block-encoding $U[P(H)]$. Supposing that $U[P(H)]$ achieves this block-encoding with $k$ ancilla qubits, they repeatedly perform the map:
\begin{align}
\ket{\psi} \to U[P(H)] \ket{0^k}\ket{\psi}
\end{align}
followed by a measurement of the $k$ ancillae. With proper interpretation of the measurement result, this approximately implements a positive operator-valued measure with operators $\{ \Pi_{\leq E}, I - \Pi_{\leq E} \}$ which forms the core of the algorithms.

Amplitude amplification is another example involving the entire block-encoding unitary. Suppose we have a projector $\Pi \approx P(A)$, and our goal is to prepare the normalized state $\Pi\ket{\psi} / \|\Pi\ket{\psi}\|$ given a unitary that prepares $\ket{\psi}$. Then an amplitude amplification algorithm could construct the required Grover reflection operators using the unitary that prepares $\ket{\psi}$ and the block-encoding unitary $U[P(A)]$~\cite{Gily_n_2019}.

Suppose that we want to enhance either of the above techniques with stochastic QSP. Now it no longer suffices to show that the quantum channel implemented by stochastic QSP approximates just the projected transformation $P(A)$. Instead, the channel must approximate $U[F(A)]$ as a whole. We find that this is indeed the case. Suppose that there is some function $G(x)$ accompanying our target function $F(x)$ that satisfies $|F(x)|^2 + (1-x^2)|G(x)|^2 = 1$ for all $x \in [-1,1]$. Then the following holds.

\begin{corollary*} 
In the setting of Theorem~\ref{thm:StochasticQSP}, consider the ensemble of quantum circuits that implement stochastic QSP (i.e. the QSP circuits that implement $\{P_j(A)\}$), upon leaving the QSP ancilla qubit(s) unmeasured. Denoting this ensemble by $\{U_j\}$, the quantum channel $\rho \to \sum_j p_j U_j \rho U_j^\dag $ approximates the map $\rho \to U[F(A)]\rho U[F(A)]^\dag$ to error $O(\epsilon)$ in diamond norm.
\end{corollary*}

\begin{proof} The precise form of $U[F(A)]$ in the most general setting for possibly non-Hermitian and possibly non-square $A$ can be reviewed in Ref.~\cite{Gily_n_2019}. In the basis presented in their Table~2, each of its matrix elements either vanish or equal $\pm F(\sigma_i) \sigma_i$ or $G(\sigma_i) \sqrt{1-\sigma_i^2}$ where $\sigma_i$ is a singular value of $A$.

By the conditions of QSP (see Eq.~\eqref{eq:qsp_conditions}), each polynomial of the ensemble $P_j(A)$ has a corresponding polynomial $Q_j(x)$ that satisfies
\begin{align}
|P_j(x)|^2 + (1-x^2)|Q_j(x)|^2 = 1
\end{align}
for all $x \in [-1,1]$. Then, for each unitary $U[P_j(A)]$ in the ensemble, the matrix elements in the same basis as considered for $U[F(A)]$ are respectively proportional to either $\pm P_j(\sigma_i) \sigma_i$ or $Q_j(\sigma_i) \sqrt{1-\sigma_i^2}$.

It follows that we can express the accuracy of the ensemble in terms of the following quantities. We define for some fixed $j$:
\begin{align}
&\alpha(x) := P_j(x) - F(x), \text{ and }  \\
&\beta(x) := Q_j(x)  - G(x).
\end{align}
Then:
\begin{align}
\Big|U[P_j(A)] - U[F(A)]\Big| \leq \sup_x \big[ \alpha(x)^2 + \beta(x)^2 (1-x^2)  \big]  \label{eqn:block_unitary_individual}
\end{align}

Now we calculate:
\begin{equation}
\begin{aligned}
1 &=  \left|F\right|^2 + (1-x^2) \left|G\right|^2\\
  &=  \left|P_j + \alpha\right|^2 + (1-x^2) \left|Q_j + \beta\right|^2\\
  &= (|P_j|^2 + \alpha P_j^* + \alpha^* P_j + |\alpha|^2)\\
  &+ (1+x^2)(|Q_j|^2 + \beta Q_j^* + \beta^* Q_j + |\beta|^2)\\
  &= 1 + (\alpha P_j^* + \alpha^* P_j + |\alpha|^2)\\
  &+ (1+x^2)(\beta Q_j^* + \beta^* Q_j + |\beta|^2),
\end{aligned}
\end{equation}
where $^*$ denotes the complex conjugate here. 
Thus:
\begin{align}
(1+x^2)|\beta Q_j^* + \beta^* Q_j + |\beta|^2| = |\alpha P_j^* + \alpha^* P_j + |\alpha|^2|.
\end{align}
We have:
\begin{equation}
\begin{aligned}
(1+x^2)|\beta Q_j^* + \beta^* Q_j + |\beta|^2| \geq |\beta|^2\\
|\alpha P_j^* + \alpha^* P_j + |\alpha|^2| \leq 2|\alpha| + |\alpha|^2 \leq 3|\alpha|
\end{aligned}
\end{equation}
So $|\beta|^2 \leq 3|\alpha| \leq 6\sqrt{\epsilon}$. Having already established that $|\alpha(x)| = |P_j(x) - F(x)| \leq 2\sqrt{\epsilon}$ in the proof of Theorem~\ref{thm:StochasticQSP}, we conclude the quantity in Eq.~\eqref{eqn:block_unitary_individual} is bounded by $\sqrt{40\epsilon}$.

Considering $\bar \alpha(x) := \sum_j p_j P_j(x) - F(x)$ and $\bar \beta(x) := \sum_j p_j Q_j(x)  - G(x)$, we observe:
\begin{align}
\Bigg| \sum_j p_j U[P_j(A)] -  U[F(A)]\Bigg| \leq \sup_x \big[ \bar\alpha(x)^2 + \bar\beta(x)^2 (1-x^2)  \big] \label{eqn:block_unitary_lincomb}
\end{align}
An identical calculation yields the bound $|\bar\beta|^2 \leq 3|\bar \alpha|$, hence bounding the quantity in Eq.~\eqref{eqn:block_unitary_lincomb} by $\sqrt{10}\epsilon \leq 4\epsilon$.

Now applying Lemma~\ref{lemma:BlockEncodingMixing} with $a = \sqrt{40\epsilon}$ and $b = 4\epsilon$, we obtain an error in diamond norm of $a^2 + 2b = 48\epsilon = O(\epsilon)$. Naturally the constant $48$ can be tightened significantly using a more careful analysis.
\end{proof}

\end{document}